\documentclass[a4paper,11pt,showpacs,amsmath,amssymb,floatfix]{article}
\usepackage[utf8]{inputenc}
\usepackage{amsmath}
\usepackage{mathrsfs}
\usepackage{amsfonts}
 \usepackage{cancel}
\usepackage{graphicx}
\usepackage{amsthm}
\usepackage{enumerate}
\usepackage{amssymb}
\usepackage{multirow} 
\usepackage{cite}  
\usepackage{bm}
\usepackage{makecell}
\usepackage{float}
\usepackage{hyperref}
 \usepackage{setspace}  
\usepackage{xcolor}
 	\definecolor{carblue}{rgb}{0.2, 0.30, 5}

\usepackage[colorinlistoftodos]{todonotes}

\usepackage{mathtools}
\mathtoolsset{showonlyrefs=true}

\newcommand{\dif}{\mathrm{d}}
\newcommand{\lr}[1]{\left( #1 \right)}
\newcommand{\lrbrace}[1]{\left\lbrace #1 \right\rbrace}

\newcommand{\lrbrkt}[1]{\left[ #1 \right]}

\newcommand{\tilg}{\widetilde{g}}
\newcommand{\Om}{\Omega}
\newcommand{\obs}{\mathcal{O}}
\newcommand{\nor}{u}
\newcommand{\tlt}{\mathring{t}}
\newcommand{\scri}{\mathscr{I}}
\newcommand{\con}{\kappa}
\newcommand{\Y}{\xi} 
 
\newcommand{\skwend}[1]{\mathrm{SkewEnd}(#1)} 
 
\newcommand{\m}{{p_a}}

\newcommand{\mink}[1]{{\mathbb{M}^{#1}}} 

\def\a{\mathfrak{a}}
\def\b{\mathfrak{b}}

\def\A{A}
\def\B{B}

\def\el{l}

\def\Ap{\mathcal{P}}

\def\barg{\overline{g}}

\def\Yv{\Y}
\def\Yf{\bm{\Y}}

\newcommand{\spn}[1]{\mathrm{span}\{ #1 \}}

\newtheorem{theorem}{Theorem}[section]
\newtheorem{proposition}{Proposition}[section]
\newtheorem{corollary}{Corollary}[theorem]
\newtheorem{lemma}{Lemma}[section]
\newtheorem{remark}{Remark}[section]
\newtheorem{definition}{Definition}[section]

\def\m{\alpha}
\def\ar{\mathrm{a}}
\def\br{\mathrm{b}}

 \title{Free data at spacelike $\scri$ and  characterization of Kerr-de Sitter in all dimensions.}

\author{
  Marc Mars and Carlos Pe\'on-Nieto\\
  Instituto de F\'{\i}sica  Fundamental y Matem\'aticas, Universidad de Salamanca \\
Plaza de la Merced s/n 37008, Salamanca, Spain}

\begin{document}

\maketitle
    
\begin{abstract} 
{
  We study the free data in the Fefferman-Graham expansion of asymptotically Einstein $(n+1)$-dimensional metrics with non-zero cosmological constant. {We analyze the relation between the electric part of the rescaled Weyl tensor at $\scri$, $D$, and the free data at $\scri$, namely a certain traceless and transverse part of the $n$-th order coefficient of the expansion $\mathring g_{(n)}$. In the case
  $\Lambda<0$ and Lorentzian signature, it was known \cite{Holl05} that
  conformal flatness at $\scri$ is sufficient for} $D$ and $\mathring g_{(n)}$ to agree up to a universal constant.   We recover {and extend this}  result to general signature and any sign of non-zero $\Lambda$. {We then explore whether conformal flatness of $\scri$ is also neceesary and link this to the validity
  of long-standing open} conjecture that no {non-trivial} purely magnetic $\Lambda$-vacuum spacetimes {exist}. In the case of $\scri$
non-conformally flat we determine a quantity constructed from an
auxiliary metric which can be used to retrieve $\mathring g_{(n)}$ from the
(now singular) electric part of the Weyl tensor.
We then concentrate in the $\Lambda>0$ case where 
the Cauchy problem at $\scri$ of the Einstein vacuum field equations is known to be well-posed when the data at $\scri$ are analytic or when
the spacetime has even dimension. 
We establish a necessary and sufficient condition for analytic data at $\scri$ to generate spacetimes with symmetries in all dimensions. These results are used to find a geometric characterization of the Kerr-de Sitter metrics in all dimensions in terms of its geometric data at null infinity.
}
    \end{abstract}

 \section{Introduction} 
  
        
  The works of R. Penrose in the 1960s  \cite{penrose63},\cite{penrose64}, \cite{penrose65} pioneered the use conformal techniques for the analysis of global properties of solutions of the Einstein equations. Eversince, conformal geometry has been widely used in general relativity. For certain conformal transformations $g = \Omega^2 \tilg$, where $\Om$ is a smooth positive function nowhere vanishing in $\widetilde M$, one can extend $g$ to a manifold with boundary $M = \widetilde M \cup \partial M $ so that the asymptotic properties of $\tilg$ become local properties at the submanifold $\scri := (\partial{M},g\mid_{\partial{M}})$, known as ``conformal infinity''.  
  

When written in terms of the conformal metric $g$, the Einstein equation of
  $\tilde{g}$ is singular at $\scri$. However, in a remarkable achievement  H. Friedrich was able (by means of introducing carefully chosen variables) to rewrite the equation in spacetime dimension four
  as a system of geometric PDE  that are regular at $\scri$
  (see the seminal works \cite{friedrich81bis}, \cite{friedrich81} and the reviews \cite{Friedrich2002}, \cite{Friedrich2014}). This system 
  allows to formulate an asymptotic initial value problem which was also proven by Friedrich \cite{Fried86initvalue} to be always well-posed if the cosmological constant is positive, a case to which we will pay special attention in this paper.

  The Cauchy problem at $\scri$ with positive cosmological constant $\Lambda$ is interesting for several reasons. First, because one can prescribe by hand the asymptotic behaviour of the spacetime. It is also noteworthy that associated to a conformal metric $g$ solving the conformal Friedrich equations, there is a solution to the Einstein equations $\tilg$ which is ``semiglobal'' (i.e. the ``physical'' spacetime $\tilg = \Om^{-2}g$ extends infinitely towards the
    future or past, depending on whether $\scri$ is a final or an initial state). 
  In addition, a remarkable simplification occurs in the constraint equations at $\scri$ as opposed to the standard constraint equations of the classical initial value problem. The data at $\scri$ consist of a conformal Riemannian three-manfiold $(\Sigma,[\gamma])$ which prescribes the (conformal) geometry of $\scri$, together with a conformal class of symmetric two-tensors $[D]$ with vanishing trace an divergence, i.e. a conformal class of transverse and traceless (TT) tensors.

 The Friedrich conformal field equations are specially taylored to dimension four and do not appear to extend to higher dimensions. The basic problem is that there do not appear to be
  enough evolution equations that remain regular at $\scri$ \cite{Friedrich2002}. Actually,
one of the fundamental objects in the conformal Friedrich equations is the rescaled Weyl tensor, which plays a central role in this paper. Our analysis shows that in dimension higher than four this object is
regular at $\scri$ only in few particular cases. Thus, there are reasons to believe that any attempt to find a regular Cauchy problem well-posed at $\scri$ based on this object will be unfruitful.

{In spite of this, one can actually find in the literature
related
  existence and uniqueness results in higher dimensions \cite{Anderson2005}
  based  on quite a different approach. This is the framework introduced by Fefferman and Graham, first in the paper \cite{FeffGrah85} and later extended into a monograph \cite{ambientmetric}. 
 An important part of the Fefferman and Graham formalism relates to asymptotically Einstein $(n+1)$-dimensional metrics, namely, conformally extendable metrics which satisfy the Einstein equations (to a certain order) at ($n$-dimensional) $\scri$. 
 This is carried through the study of their asymptotic formal series expansions,
 (see  also  \cite{Starob82} for an earlier asymptotic study in four spacetime dimensions and positive cosmological constant)
 usually called Fefferman-Graham (FG) expansion. It is remarkable the qualitative differences that appear between the $n$ odd and even cases.  For $n$ even, there is an obstruction to find a smooth metric satisfying the Einstein equations to infinite order at $\Sigma$. The expansion introduces
logarithmic terms which depend on a certain conformally invariant tensor $\obs$ determined by the boundary metric, known as obstruction tensor. No such obstruction occurs for $n$ odd. 
 Interestingly it is the obstruction tensor what allows Anderson \cite{Anderson2005} to find an asymptotic Cauchy problem for the Einstein equations in the $n$ odd case.  We remark that nor the original Anderson's proof in  \cite{Anderson2005}, neither Anderson and Chruściel's later clarification \cite{andersonchrusciel05} are fully complete and a revised complete version has been recently published in \cite{kaminski21}.  For $n+1$ even dimensional metrics $\tilg$, this tensor provides the differential equation $\obs = 0$, which for  Lorentzian conformally Einstein metrics, can be cast as a Cauchy problem at $\scri$.
 Anderson proves that solutions of this Cauchy problem exist and are uniquely determined for every pair of symmetric two-tensors $(\gamma, g_{(n)})$, $\gamma$ positive definite and $g_{(n)}$ traceless and transverse w.r.t. $\gamma$. A posteriori, $\gamma$ determines the geometry of $\scri$ and $g_{(n)}$ is $n$-th order coefficient of the asymptotic expansion of $\tilg$. This idea is not extendable to the $n$ even case, for no obstruction tensor can be built out of $\tilg$ when $n+1$ is odd. In this case, however, there is a result \cite{kichenassamy03} that proves the convergence of the FG expansion in the case where the data are analytic.
 }

 Conversely, an existence and uniqueness theorem can be used to characterize spacetimes by means of their Cauchy data. The situation is particularly interesting in the
 case of the asymptotic Cauchy problem for positive $\Lambda$, because of the simplicity of the data, which potentially allows one to achieve
 classification results for spacetimes whose explicit form need not to be known.

 The data for the existence and uniqueness theorems of dimension higher than four are the coefficients of the Fefferman-Graham expansion $(\gamma,g_{(n)})$. The original definition of the coefficient $g_{(n)}$ is not covariant, because the Fefferman-Graham expansion is constructed in a very particular set of coordinates, which in general is not easily obtainable. By using Hamiltonian methods, a fully covariant characterization of $g_{(n)}$ was presented in \cite{papasken04bis,papasken04}. The method relies on the fact that the canonical momenta can be obtained as functional derivatives of the on-shell gravity action, together with expansions in terms of eigenfunctions of the dilatation operator.  In the first part of the paper we follow a completely different route to characterize geometrically $g_{(n)}$ in the special case that $\scri$ is conformally flat.  In such case, we  provide a geometric relation between $g_{(n)}$ and the tensor $D$, the electric part of the rescalled Weyl tensor at $\scri$, for spacetimes admitting a conformally flat $\scri$. Specifically, $g_{(n)}$ which in all cases has trace and divergence  determined by $\gamma$, can be easily split in the conformally flat case into a term explicitly determined by $\gamma$ and a traceless and divergence-free term $\mathring g_{(n)}$ which, up to a constant, agrees with $D$.  This extends to the $\Lambda>0$ case a previous result by Hollands-Ishibashi-Marolf \cite{Holl05} in the $\Lambda<0$ case. Actually, this extension is straightfoward if one takes into account general results \cite{Sken02} relating the coefficients of the Fefferman-Graham expansion for opposite signs of $\Lambda$. Our contribution is therefore not new with regard the identification of $\mathring g_{(n)}$ and  $D$, but it goes beyond in several aspects. Firstly, we work in arbitrary conformal gauge and make no restriction on the signature nor the sign of (non-zero) $\Lambda$. Secondly, we derive the result as a simple corollary of a general identity for the Weyl tensors of two metrics related in  a certain way (see Lemma \ref{lemmaWeyls}). This general identity is of independent interest. Finally we discuss whether conformal flatness is not only necessary for the coincidence of $\mathring g_{(n)}$ and $D$ (up to a constant), but also sufficient.   We conclude 
   that this is the case as long as no purely magnetic spacetimes satisfy the $\Lambda$-vacuum Einstein equation. Interestingly, this is a long standing open conjecture in general relativity, see e.g. \cite{mcintosh94} and also \cite{barnes04}. Finally, we give ideas for relating $D$ and (the total coefficient) $g_{(n)}$ in the non-conformally flat $\scri$ case. We show that $D$ is in general divergent
and find an explicit expression involving the Weyl tensor of an auxiliary metric that needs to be subtracted  to $D$ in order to retrieve $g_{(n)}$.


 A characterization result of spacetimes via data at $\scri$ with positive $\Lambda$ must encode all the information of the corresponding spacetime. In particular, the presence of symmetries must also constraint these data. The case $n=3$ has been studied in \cite{KIDPaetz}
 where it is shown that the spacetime admitting  a Killing vector
 is equivalent to the TT tensor $D$ satisfying geometrically neat equation, the so-called Killing initial data (KID) equation. This equation involves a conformal Killing vector field (CKVF) of $\gamma$ which is a posteriori the Killing vector field at $\scri$ in the conformally extended spacetime. Apart from this $n=3$ case, no previous results relating continuous local isometries to initial data at $\scri$ were known in more dimensions. In this paper
 we prove a higher dimensional result, analogous to the $n=3$ one, restricted to the case of analytic metrics with zero obstruction tensor.

 The results described above are used in the second part of the paper to characterize the generalized Kerr-de Sitter spacetimes \cite{Gibbons2005}.
 We find that the data corresponding to Kerr-de Sitter is characterized
   by the conformal class of $\gamma$ being locally conformally flat and the
   tensor $\mathring g_{(n)}$ taking  the form
   $\mathring g_{(n)} = D_\Y$  where $D_\Y$ is a TT tensor depending on a conformal Killing vector field (CKVF) $\Y$ of $\gamma$. More concretely, the CKVF must belong to a specific conformal class  which we also explicitly determine. Since by the results in the first part of the paper $\mathring g_{(n)}$ admits a clear geometric
   interpretation in terms of the rescalled Weyl tensor whenever
   $\gamma$ is locally conformally flat, our characterization result of Kerr-de Sitter at $\scri$ is fully geometric. The result here is a 
   natural generalization of the already known case for $n=3$, studied in \cite{KdSlike} (see \cite{KdSnullinfty} for the non-conformally flat $\scri$ case,
   and also \cite{gasperinwilliams17}).

 The contents of the paper are organized as follows. In Section \ref{secweylgn} we start by setting the basics of the Fefferman-Graham formalism. Some general results of conformal geometry are also proven and two useful formulas for the Weyl tensor are derived (cf. Lemmas \ref{lemmaCt} and \ref{lemmaWeyls}, with proofs in Appendix \ref{appA}), which have several applications in this paper,
 an example being the calculation of the initial data in the second part of the paper. Then Theorem \ref{theognweyl} is proven, which establishes
 that the $n$-th order coefficient of the FG expansion coincides (up to a certain constant) with the electric part of the rescaled Weyl tensor in the
   case when $\scri$ is conformally flat and $n>3$ (for $n=3$ this is true in
 full generality). This theorem finds immediate application in the Cauchy problem of Einstein equations at $\scri$ with positive cosmological constant (cf. Corollary \ref{coroland}). In Section \ref{appendixKID} we derive a KID equation for analytic data at $\scri$ for $n$ odd and even with vanishing obstruction tensor (we indicate that the result should also hold when
  the  obstruction tensor is non-zero, but this requires additional analysis). This equation is necessary and sufficient for the Cauchy development of the data at $\scri$ to admit a Killing vector field. Our final Section \ref{secKdsmet} gives an interesting application of the previous results. Namely, we calculate the initial data of the Kerr-de Sitter metrics in all dimensions \cite{Gibbons2005}. As mentioned above, these data are uniquely characterized by the conformal class of a CKVF $\Y$. In order to give a complete characterization, we identify this conformal class using the results in \cite{marspeon21}.

Throughout this paper, $n$ refers to the dimension of $\scri$ and $n+1$ the dimension of the spacetime. Several results below are valid in arbitrary signature
and arbitrary sign of the cosmological costant. We will work in the general setup unless otherwise stated.

\section{Initial data and the Weyl tensor}\label{secweylgn}

In this section we relate the initial data at spacelike conformally flat $\mathscr I$ which appears in Anderson's and Kichenassamy existence and uniqueness theorem \cite{Anderson2005}, \cite{kichenassamy03} to the electric part of the rescaled Weyl tensor. This theorem relies on the Fefferman-Graham (FG) expansion of Poincaré metrics near $\mathscr I$. These are a generalization of the Poincaré metric of the disk model of the hyperbolic space, whose conformal infinity is given by the conformal structure associated to the usual $n$-sphere. The properties concerning Poincaré metrics that are needed in this paper will be stated next, and we refer to the original publication \cite{FeffGrah85} and to the extended monograph \cite{ambientmetric} for further details.

 \subsection{Fefferman-Graham formalism}

Consider an $(n+1)$-dimensional pseudo-Riemannian manifold $(M,g)$ with boundary $\partial M$ and denote its interior by $\widetilde M = \mathrm{Int}(M)$. 
Then $(M,g)$ is said to be a {\it conformal extension } of $(\widetilde M, \tilg)$ if there exists a smooth function $\Om$ positive on $\widetilde M$ such that 
\begin{equation}
g = \Om^2 \tilg \quad\quad \mbox{and}\quad\quad \partial M = \{\Om = 0 \cap \dif \Om \neq 0 \}.
\end{equation} 
%
%
The extended metric $g$ is assumed to be smooth in $\widetilde{M}$ but only to have finite differentiability up to $\partial M$. When $g$ is smooth also at the boundary, we will call this a {\it smooth conformal extension}. The submanifold $\scri := (\partial M, g|_{\partial M})$, which will be non-degenerate of signature $(p,q)$ in the cases we shall deal with, is called ``conformal infinity''. Notice that multiplying the conformal factor $\Om$ by any smooth positive function $\widehat \omega$ yields a different conformal extension such that $\scri = (\partial M, \omega^2 g|_{\partial M})$, where $\omega = \widehat\omega|_\scri$. Hence one usually considers $\scri$ as $\partial M$ equipped with the whole conformal class of metrics $[g|_{\partial M}] = \omega^2 g|_{\partial M},~\forall \omega \in C^\infty(\partial M),~\omega>0$. A metric admitting a conformal extension is said to be conformally extendable. Also, throughout this paper by ``conformally flat'' we mean ``locally conformally flat'' unless otherwise stated.  

The so-called Poincaré metrics associated to a conformal manifold, i.e. a smooth manifold endowed with a conformal structure $(\Sigma,[\gamma])$, are metrics $\tilg$ admitting a smooth conformal extension $g = \Om^2 \tilg$ with prescribed conformal infinity $\scri :=(\Sigma,[\gamma])$. The extension $g$ is defined in an open neighbourhood $M$ of  $\Sigma \times \{0\}$ in $\Sigma \times [0, \infty)$ (and $\tilg$ in $ \widetilde M = M - \{ \Sigma \times \{0\} \}$) satisfying the following conditions. $\Sigma$ is naturally embedded in $M$ by $i: \Sigma \hookrightarrow \Sigma \times [0,\infty)$, where $i(\Sigma) = \Sigma \times \{0\}$, so we  identify $\Sigma$ and $\Sigma \times \{0 \}$ when there is no risk of confusion. In $M$, $\Om$ is a defining function of $\Sigma  =  \{ \Om  = 0 \}$ as before.
 Moreover, the definition imposes, depending on the parity of $n$, a certain decay rate of the tensor $Ric(\tilg) - \lambda n \tilg$ near $\scri$, where $Ric(\tilg)$ is the Ricci tensor of $ \tilg$ and $\lambda:= \frac{2}{n(n-1)} \Lambda$  with $\Lambda$ the cosmological constant, which we assume to be non-zero. Following \cite{FeffGrah85} we say that a symmetric $2$-tensor field $D$ is  $O^+(\Om^m)$ if it is $D = O(\Om^m)$ and $\mathrm{Tr}_\gamma i^\star(\Om^{-m}D|_{\scri}) = 0$.
With this notation we can give the formal definitions \cite{FeffGrah85}, \cite{ambientmetric}:
 \begin{definition}  \label{defpoincare}
  A Poincaré metric for a conformal $n$-manifold $(\Sigma,[\gamma])$
  of signature $(p,q)$, is a metric $\tilg$ of signature $(p+1,q)$ if $\lambda>0$ or $(p+1,q)$ if $\lambda<0$ admitting a smooth conformal extension such that $\scri = (\Sigma, [\gamma])$ and
  \begin{enumerate}
   \item If $n = 2$ or $n \geq 3$ and odd, $Ric(\tilg) - \lambda n \tilg$ vanishes to infinite order at $\Sigma$.
   \item If $n \geq 4$ and even, $Ric(\tilg) - \lambda n \tilg $ is $ O^+(\Om^{n-2})$.
  \end{enumerate}
 \end{definition}

 Let $\tilg$ be a conformally extendable metric $\tilg$  and $g$ a conformal extension. From the well-known relation between Ricci tensors two conformal metrics  (e.g. 
 \cite{Kroonbook}) 
 \begin{equation}\label{eqrelriccis}
  R_{\alpha \beta} - \widetilde R_{\alpha \beta} = - \frac{n-1}{\Om} \nabla_\alpha \nabla_\beta \Om - g_{\alpha \beta} \frac{\nabla_\mu \nabla^\mu \Om}{\Om} + g_{\alpha \beta}\frac{n}{\Om^2} \nabla_\mu \Om \nabla^\mu \Om  
 \end{equation}
 it follows \cite{GrahamLee91} that $\tilg$ satisfies the Einstein equation $Ric(\tilg) - \lambda n \tilg = O(\Om^{-1})$ if and only if 
   $(\Om^{-2}\tilg^{\alpha \beta} \nabla_\alpha \Om  \nabla_\beta \Om )|_{\scri} = (g^{\alpha \beta} \nabla_\alpha \Om  \nabla_\beta \Om )\mid_{\scri} = -\lambda$. If this holds, it can be proven \cite{mazzeo88} that all sectional curvatures of $\tilg$ take a limit at $\partial M$ equal to $\lambda$. In the Riemannian case (which requires $\lambda <0$) these metrics are called {\it asymptotically hyperbolic}. In the case with general signature and arbitrary non-zero $\lambda$ we call them {\it asymptotically of constant curvature} (ACC). It is obvious from Definition \ref{defpoincare} that  Poincaré metrics are ACC. We define then:
 
\begin{definition}\label{defnormf}
Let $\tilg$ be an ACC and conformally extendable metric with  $\scri = (\Sigma,[\gamma])$. Let also be a representative $\gamma \in [\gamma]$. Then $\tilg$, as well as the conformally extended metric $g = \Om^2 \tilg$,  are said to be in normal form w.r.t. $\gamma$ if 
  \begin{equation}\label{eqnormform}
 \tilg =\frac{1}{\Om^2} (-\frac{\dif \Om^2}{\lambda} + g_\Om),\quad\quad g = \Om^2 \tilg = -\frac{\dif \Om^2}{\lambda} + g_\Om,
\end{equation}
where  $g_{\Om}$ is a family of induced metrics on the leaves $\Sigma_\Om = \{ \Om = const.\}$ such that $g_{\Om}\mid_\Sigma = \gamma$.
\end{definition}
Notice that it is always possible to adapt coordinates $\{\Om, x^i \}$ to the normal form \eqref{eqnormform}, where $\partial_\Om:= \partial/\partial \Om$ and $\partial_i := \partial/\partial_{x^i}$ are normal and tangent to the $\Sigma_\Om$ leaves respectively. We denote these as {\it Gaussian coordinates}\footnote{In the usual definition of Gaussian coordinates, the $\partial_\Om$ vector is unit. The introduction of a constant factor $\lambda$ does not modify its general properties.}. 
For an ACC metric $\tilg$ which satisfies the Einstein equations to infinite order at $\scri$, the Fefferman-Graham expansion (FG) is a formal expansion in Gaussian coordinates
\begin{align}
 g_\Om  & \sim \sum\limits_{s = 0}^{(n-1)/2} g_{(2 s)}  \Om^{2s}  + \sum\limits_{s =n}^{\infty} g_{(s)}  \Om^s,\quad\quad &&\mbox{if $n$ is odd},\label{eqFGexpoddsec}\\
  g_\Om  & \sim \sum\limits_{s = 0}^{\infty} g_{(2 s)}  \Om^{2s}  + \sum\limits_{s = n/2}^{\infty} \sum\limits_{t = 1}^{m_s} \obs_{(2s,t)} \Om^{2s} (\log \Om)^t ,\quad\quad &&\mbox{if $n$ is even}\label{eqFGexpevsec},
\end{align}
where $m_s \leq 2s-n+1$ in an integer for each $s$, the coefficients $g_{(s)}$ are objects defined at $\scri$ and extended to $M$ as independent of $\Om$ and the logarithmic terms arise in the $n$ even case whenever the so-called obstruction tensor $\obs_{(n,1)} = \obs$ for $\gamma$ is non-zero. This tensor is a conformally invariant symmetric and trace-free 2-covariant in even dimensions and its vanishing is a necesssary and sufficient condition for the tensor $g_{\Om}$ to be smooth up to and including $\scri$. 




%
%
%
%
%
The appearance of the logarithmic terms (whenever the obstruction tensor is non-zero) is a consequence of forcing the metric to satisfy the Einstein equations to infinite order at $\scri$ with $n$ even. These terms spoil smoothness. Thus, in order to distinguish this case from the Poincaré metrics, which are smooth by  Definition \ref{defpoincare}, metrics  satisfying the Einstein equations  with non-zero cosmological constant to infinite order at $\scri$ are called {\bf Feffeman-Graham-Poincaré (FGP)} in this paper. Obviously, whenever  $n$ is odd or $n$ is even and $\obs =0$, a FGP metric is Poincaré.

For FGP metrics, the coefficients of the FG expansions are generated by recursive relations obtained from the Einstein equations at $\scri$ as follows. One must first prescribe the zero-th order coefficient $g_{(0)}$, which in turn determines the boundary metric $g_{(0)} = \gamma$. Then, the recursive relations give each coefficient of order $s \neq n$, as a function of previous terms and tangential derivatives up to order $s-2$.  As a consequence, the expansions \eqref{eqFGexpoddsec} and \eqref{eqFGexpevsec} are both even up to order strictly lower than $n$ and the terms $g_{(s)}$ for $s<n$ are uniquely generated by $\gamma$.  Explicit expressions for the first terms  $g_{(s)}$ as well as for the obstruction tensor up to dimension $n=6$ are obtained in 
\cite{deHaro} (the complexity of the respective expressions increases rapidly with $s$ and with the dimension). The $n$-th order coefficient is independent of $\gamma$, except for its trace and divergence,
 \begin{equation}\label{eqTTgsec}
 \mathrm{Tr}_\gamma g_{(n)} := \a(\gamma), \quad\quad \mathrm{div}_\gamma g_{(n)} := \b(\gamma),
\end{equation} 
where $\a(\gamma) = 0,~\b(\gamma) = 0$ for $n$ odd and $\a(\gamma)$ is a scalar and $\b(\gamma)$ a one-form covariantly determined by $\gamma$ for $n$ even.  Also, $\obs$ is obtainable from $\gamma$ and in particular, it vanishes for conformally flat $\gamma$, a case  which will be central in this paper. The subsequent coefficients $g_{(s)}$ (and also $\obs_{(s,t)}$ whenever they appear) with $s > n$ are obtained by previous coefficients up to order $g_{(s-2)}$ (and $\obs_{(s-2,t)}$, whenever they appear).

 In other words, the freely prescribable data at $\Sigma$, which determine a unique FGP metric $g$ to infinite order at this hypersurface, are the zero-th order coefficient, which is given by the boundary metric $\gamma$, and a symmetric two-covariant tensor $h$, whose trace and divergence are constrainted by $\a(\gamma),~\b(\gamma)$ respectively, and which determines the $n$-th order coefficient.  Stated as a formal Theorem:

 \begin{theorem}[Fefferman-Graham \cite{ambientmetric}]\label{theoFG}
Let $(\Sigma,\gamma)$ be a pseudo-Riemannian manifold of signature $(p,q)$ and let $h$ be a symmetric two-tensor of $\Sigma$. 
\begin{itemize}

\item If $n = 2$ and if $\mathrm{div}_\gamma h =\frac{1}{2\lambda} \dif (\mathrm{Scal}_\gamma)$ and $\mathrm{Tr}_\gamma h =\frac{1}{\lambda}  \mathrm{Scal}_\gamma$, there exists an even (i.e. with only non-zero coefficients of even order) Poincaré metric $g$ in normal form w.r.t $\gamma$ which admits an expansion of the form \eqref{eqFGexpevsec} (with $\obs_{(2s,t)} = 0$) and $g_{(2)} = h$.
 \item If $n\geq 3$ is odd and  if $\mathrm{div}_\gamma h = 0$  and $\mathrm{Tr}_\gamma h = 0$, there exists a Poincaré metric $g$ in normal form w.r.t $\gamma$, which admits an expansion of the form \eqref{eqFGexpoddsec} such that $g_{(n)} =  h$ (in particular, trace-free). 
\item If  $n\geq 4$ is even, there exist a one-form $\b(\gamma)$ and a scalar $\a(\gamma)$, both covariantly determined by $\gamma$, is such a way that if $\mathrm{div}_\gamma h = \b(\gamma)$ and $\mathrm{Tr}_\gamma(h) = \a(\gamma)$, then there exists a FGP which admits an expansion of the form \eqref{eqFGexpevsec} such that $g_{(n)} = h$. The solution is smooth if and only if the obstruction tensor of $ \gamma$ vanishes.
\end{itemize}
 \end{theorem}

Away from $\Sigma$ there are multiple ways of extending $g$ if no further assumptions are made.  Anderson's existence and uniqueness theorem \cite{Anderson2005} (see also  \cite{andersonchrusciel05} and \cite{kaminski21} for a fully complete proof) proves that for metrics of Lorentzian signature and $n$ odd, there is a unique way of extending this metric away from $\Sigma$ so the $\Lambda>0$ vacuum Einstein field equations hold, namely
\begin{equation}\label{EFE}
\widetilde R_{\alpha \beta}  = n \lambda \tilg_{\alpha \beta}.
\end{equation} 
 In other words, the coefficients $(\gamma,g_{(n)})$, with $(\Sigma,\gamma)$ Riemannian, determine a unique Einstein metric for $(\Sigma, \gamma)$ in a collar neighbourhood of $\Sigma$.  Thus, a triple $(\Sigma,\gamma,g_{(n)})$, where $g_{(n)}$ satisfies the conditions of Theorem \ref{theoFG}, will be called {\bf asymptotic data}. The proof of Anderson's theorem relies on imposing that the obstruction tensor of $g$ vanishes. This can be done in the $n$ odd case because $g$ is $n+1$ (even) dimensional. This method is not applicable to the $n$ even case, for no obstruction tensor is associated to $g$ when $n+1$ is odd. 
 
 

\begin{theorem}[Anderson, Chruściel, Kamiński  \cite{Anderson2005,andersonchrusciel05,kaminski21}]\label{theoanderson}
   Let $n \geq 3$ odd. For every choice of asymptotic data $(\Sigma, \gamma, g_{(n)})$, with $\gamma$ Riemannian, there is a unique Lorentzian metric solving Einstein's equations for $\Lambda>0$ in a neighbourhood of $\scri$. This problem is well-posed in suitable Sobolev spaces.
 \end{theorem}
 
 {
 Next, we shall comment on the convergence results of the FG  expansion of a FGP metric in the case when the data
    $(\gamma,g_{(n)})$ at $\scri$ are analytic.  The most general convergence result for this case is obtained by Kichenassamy in \cite{kichenassamy03}. This reference assumes that $\gamma$ is positive definite, but no sign of $\Lambda$ is imposed. The convergence is proven {for arbitrary $n$}, including the case where non-zero logarithmic terms appear ({only  when} $n$ is even).
    Actually, the convergence guarantees that the Einstein equations are in fact solved in a sufficiently small
  neighbourhood of $\Om = 0$, not just to infinite order at $\scri$. Hence, this implies existence and uniqueness of Einstein metrics in a neighbourhood of $\scri$ generated by analytic data $(\gamma,g_{(n)})$. These will be called {\bf asymptotic data in the analytic class}. We also note that the convergence in the $n$ odd case with $g_{(n)} = 0$ was already established in \cite{FeffGrah85} in full generality, i.e. with no restrictions on the signature of $\gamma$ nor the sign
    of $\lambda$. 
    The next Theorem summarizes the results that we shall require in the analytic case.
    \begin{theorem}[Kichenassamy \cite{kichenassamy03}]\label{theokichen}
      Let $n \geq 1$. For every choice of asymptotic data in the analytic class $(\Sigma, \gamma, g_{(n)})$, with $\gamma$ Riemannian, the FG expansion converges in a neighbourhood of $\scri$ to a unique Lorentzian (resp. Riemannian) metric solving the Einstein equations for $\Lambda>0$ (resp. $\Lambda <0$).
 \end{theorem}
    }
    
%
%

In this paper, a conformally extended metric (i.e. defined at $\scri$) will be denoted by $g$ and the metric from which it is extended by $\tilg$ (i.e. not defined at $\scri$). We use tilde to distiguish geometric objects associated to $\tilg$ from those associated to $g$. For instance, $\widetilde \nabla$ is the Levi-Civita connection of $\tilg $ and $\nabla$ the one of $g$. The signature of the metric $\tilg$ and the sign of $\lambda$ will remain general, namely $(p+1,q)$ if $\lambda>0$  or $(q,p+1)$ if $\lambda<0$ unless otherwise specified. However, we note that our main interest is in the positive cosmological constant setting, specially for the applications in the second part of the paper. Therefore, for the sake of simplicity some of our results are given in detail only for this case (namely Lemma \ref{lemmapoinconflat}, Proposition \ref{propconflatdec} and Theorem \ref{theognweyl}), while the negative $\lambda$ will be just indicated.
%

 If $g$  is an ACC metric in normal form w.r.t. a representative $\gamma \in [\gamma]$, it is immediate to verify that the vector field $T := \nabla \Om$ is geodesic (affinely parametrized). 
 Any conformal extension such that $ T$ is geodesic is called a {\it geodesic conformal extension}. We start by proving some general results about this kind of extensions. In the following, Greek indices $\alpha,\beta$ running form zero to $n$ are used for spacetime coordinates. For spacelike hypesurfaces (usually $\{ \Om = const.\}$) we use Latin indices $i= 1, \cdots, n$.

\begin{lemma}\label{lemmageodesic}
Let $\tilg$ be an ACC metric. Then, a conformal extension $g = \Om^2 \widetilde g$ is geodesic if and only if
 \begin{equation}
  \nabla_\alpha \Om \nabla^\alpha \Om =-\lambda .
 \end{equation}
\end{lemma}
\begin{proof}
 The lemma follows from
 \begin{equation}\label{eqmodnabome} 
  \nabla^\beta \Om \nabla_\beta \nabla_\alpha \Om  =  \nabla^\beta \Om \nabla_\alpha \nabla_\beta \Om = \frac{1}{2}\nabla^\alpha \Om \nabla_\alpha \lr{\nabla_\beta \Om \nabla^\beta \Om } 
 \end{equation}
 because if $\nabla_\alpha \Om \nabla^\alpha \Om = - \lambda$ the RHS of \eqref{eqmodnabome} vanishes and $T$ is geodesic and, conversely, if $T$ is geodesic then \eqref{eqmodnabome} is zero, so $\nabla_\beta \Om \nabla^\beta \Om$ is constant and ,$g$ being ACC, its value is everywhere equal to 
$-\lambda$.
\end{proof}

The following result guarantees the existence of geodesic conformal extensions for each choice of boundary metric $\gamma$. The proof, which we detail next, is similar to the one for $\Lambda< 0$, which appears in \cite{GrahamLee91} (Lemma 5.2). Before stating the lemma, let us briefly review the non-characteristic condition for a first order PDE Cauchy problem (see \cite{evanspdes}, Chapter 3). For this it will be convenient to use coordinates $\{x^\alpha\} = \{\Om, x^i\}$ adapted to the initial hypersurface, that is, taking coordinates $\{ x^i \}$ at $\Sigma =  \{ \Om = 0 \}$ and propagating them as $T^\alpha\partial_\alpha x^i = 0$.
Consider a first order PDE Cauchy problem 
\begin{equation}\label{eqcauchyp}
 F(x^\alpha; u, \nabla_\alpha u) = 0,\quad\quad u\mid_\Sigma = \phi,
\end{equation}
where $u$ is a scalar function.
Two functions $\{\phi, \psi_0\}$  of  $\Sigma$ are a set of {\it admissible} initial data whenever they satisfy the following {\it compatibility condition}
\begin{equation}\label{eqwelldata}
 F(x^0 = 0,x^i;\phi,\psi_0,\frac{\partial \phi}{\partial x^1},\cdots,\frac{\partial \phi}{\partial x^n}) = 0.
\end{equation}
Denote  $\mathcal{D}_{\nabla_\alpha u} F$ to the derivative of $F$ w.r.t. $\nabla_\alpha u$ and let $V(x^\alpha;u,\nabla_\alpha u)$ be the vector of components $V^\alpha = \mathcal{D}_{\nabla_\alpha u} F $. Also, let $T$ be the normal covector to $\Sigma$, i.e. $T_\alpha = \nabla_\alpha \Om$. Then, for every set of admissible initial data, the Cauchy problem is said to be {\it non-characteristic} if
\begin{equation}
 T \cdot V(x^0 = 0,x^i;\phi,\psi_0,\frac{\partial \phi}{\partial x^1},\cdots,\frac{\partial \phi}{\partial x^n})  \neq 0,
\end{equation}
where $\cdot$ denotes the usual action of a covector on a vector.
A non-characteristic Cauchy problem is known to be locally well-posed (e.g.  \cite{evanspdes}), i.e. that there exists a solution $u$ of \eqref{eqcauchyp}, satisfying $u|_\Sigma = \phi,~\partial_0 u |_{\Sigma} = \psi_0$. 
After this remark, we can prove the next Lemma.

\begin{lemma}\label{lemmaexistgeod}
 Let $ \widetilde g$ be an ACC metric with conformal infinity $(\Sigma, [\gamma])$. Then, for every representative $\gamma \in [\gamma]$, there exist a  geodesic conformal extension $g = \Om^2 \widetilde g$ which induces the metric $\gamma$ at $\Sigma$.  
\end{lemma}
\begin{proof} 
 Consider a conformally extended metric $g$ such that $g = \Om^2 \widetilde g$ and $g\mid_{\Om = 0} = \gamma$. Let $\hat g \in [\widetilde g]$ be such that $\hat g = \omega^2 g$ with $\omega>0$ and $\omega \mid_{\Om=0} = 1$, so that $\hat g$ realizes the same boundary metric $\gamma$. Therefore $\hat g = \hat \Om^2 \widetilde g$, with $\hat \Om = \omega \Om$, so by Lemma \ref{lemmageodesic}, we have to show that there exists a function $\omega$ such that $\hat \Om$ satisfies \eqref{eqmodnabome} for the metric $\hat g$
 \begin{equation}
  \hat g^{\alpha \beta} \nabla_\alpha \hat \Om \nabla_\beta \hat \Om = \frac{g^{\alpha\beta}}{\omega^2}\nabla_\alpha(\omega \Om) \nabla_\beta(\omega \Om) = -\lambda.
 \end{equation}
Expanding the derivatives and defining $u := \log \omega$, this amounts to
 \begin{equation}\label{eqnoncharac}
   g^{\alpha \beta}\lr{2 \nabla_\alpha \Om \nabla_\beta u + \Om \nabla_\alpha u \nabla_\beta u } = \frac{-\lambda - g^{\alpha \beta} \nabla_\alpha \Om \nabla_\beta \Om}{\Om}.
 \end{equation}
 The LHS of  \eqref{eqnoncharac} is obviously regular at $\Om = 0$. Also, since $g$ is ACC $g^{\alpha \beta} \nabla_\alpha \Om \nabla_\beta \Om\mid_{\Om = 0} = -\lambda$, thus the RHS tends to $-\partial_{\Om}\lr{g^{\alpha \beta} \nabla_\alpha \Om \nabla_\beta \Om}$ at $\Om = 0$, which is regular. Hence, we can pose a Cauchy problem at $\{\Om = 0\}$, for which
 we must complete $\phi = \log \omega \mid_\Sigma = 0$ to admissible initial data for \eqref{eqnoncharac}. These data must satisfy \eqref{eqwelldata}, thus $\psi_0$ is fixed to satisfy
 $2 g^{0 0} \psi_0 = -\partial_\Om(g^{\alpha \beta}\nabla_\alpha \Om \nabla_\beta \Om)\mid_\Sigma$. The vector field $V$ has components  $2 g^{\alpha\beta} ( \nabla_{\beta} \Omega + \Omega \nabla_{\beta} u)$ and therefore
\begin{equation}
 T \cdot V(x^0 = 0,x^i;\phi,\psi_0,\frac{\partial \phi}{\partial x^1},\cdots,\frac{\partial \phi}{\partial x^n})    = 2 g^{\alpha\beta} \nabla_\alpha \Om \nabla_\beta \Om\mid_{\Sigma} = -2 \lambda.
\end{equation}
Hence the problem is non-characteristic if $\lambda \neq 0$. 
\end{proof}
\begin{remark}\label{remarkPoinnorm}
 This lemma combined with the use of Gaussian coordinates $\{\Om,x^i\}$ 
implies that given an ACC metric $\tilg$ and any representative $\gamma$  of its boundary conformal structure,  there exists coordinates near $\partial M$ where $\tilde{g}$ is written in normal form w.r.t that representative (see Definition \ref{defnormf}).
\end{remark}


\subsection{Formulae for the Weyl tensor}

We next state some useful results in order to calculate the Weyl tensor.Lemma \ref{lemmaCt} below gives the {\it $T$-electric part of the Weyl tensor} in Gaussian coordinates, which we define by
\begin{equation}\label{eqTelec}
 \lr{C_{T}}_{ \alpha \beta} := {C}_{\mu \alpha \nu \beta} T^\mu T^\nu = \Om^2 \widetilde C_{\mu \alpha \nu \beta} T^\mu T^\nu.
\end{equation}
 This tensor is proportional (with non-zero factor) to the standard {\it electric part of the Weyl tensor} w.r.t. to a unit vector $\nor$
\begin{equation}\label{eqelec}
 \lr{C_\perp}_{ \alpha \beta} := {C}_{\mu \alpha \nu \beta} \nor^\mu \nor^\nu.
\end{equation}
 For geodesic conformal extensions, the proportionality is just a constant, namely $C_\perp = \lambda^{-1} C_T$ provided $T$ and $u$ point into the same direction. If, in addition, the metric is ACC the rescaled Weyl tensors always satisfy
 \begin{equation}\label{CtCperp}
  (\Om^{2-n}C_\perp)\mid_\scri = \lambda^{-1} (\Om^{2-n} C_T)\mid_\scri,
 \end{equation} 
  whenever these quantities are finite. Hence, by adding the constant factor $\lambda$ we can use interchangeably the electric and $T$-electric parts of the Weyl tensors at $\scri$.

 After this remark, we state one of the two lemmas of this subssection, whose proof is given in Appendix \ref{appelec}.

\begin{lemma}\label{lemmaCt}
 Let $\tilg$ be a conformally extendable Einstein metric with $\Lambda\neq 0$ and $g = \Om^2 \tilg$ a geodesic conformal extension. Then, in Gaussian coordinates $\{\Om, x^i \}$, the $T$-electric part of the Weyl tensor reads
 \begin{equation}\label{eqCtnocov}
 \lr{C_{T}}_{ij} = \frac{\lambda^2}{2} \lr{\frac{1}{2} \partial_{\Om} g_{ik} g^{kl} \partial_{\Om} g_{lj}+  \frac{1}{\Om} \partial_\Om g_{ij} - \partial^{2}_\Om g_{ij}}.
\end{equation}
where $g_{ij} = (g_\Om)_{ij}$ is the metric induced by $g$ on the leaves $ \{\Om = const. \}$.
\end{lemma}
\begin{remark}\label{remarkg3}
 Note that equation \eqref{eqCtnocov} implies that $C_T$ is always $O(\Om)$. In particular, in dimension $n = 3$ it is always the case that 
 \begin{equation}
  (\Om^{-1}C_T)\mid_\scri = -\frac{3 \lambda^2}{2} g_{(3)}
 \end{equation}
which recovers the well-known result by Friedrich \cite{Fried86initvalue}  that for positive  $\Lambda$ the electric part of the rescaled Weyl tensor corresponds to the free data specifiable at $\scri$.
\end{remark}

 Assume that $\hat g $ satisfies the hypothesis of Lemma \ref{lemmaCt} and that {its FG expansion is of the form }
    \begin{align*}
    \hat g  \sim \sum_{s=0}^{(n-1)/2} g_{(2s)} \Omega^{2s} + \Omega^{n+1} \el
  \end{align*}
  with $n$ odd and $\el$ at least $C^2$ up to an including $\{\Omega=0\}$.
  Equation \eqref{eqCtnocov} implies that its $T$-electric Weyl tensor
  $\hat C_T$ only has even powers of $\Omega$ up to and including $\Omega^{n-1}$
  (higher order terms may be even and odd).
  As a consequence, the tensor $\Omega^{2-n} \hat C_T$ splits as a sum of divergent terms at $\Omega=0$ plus a regular part which vanishes at
    $\Omega=0$. This fact will have interesting consequences when combined with Lemma \ref{lemmaWeyls} below.

    We now present a general result concerning the Weyl tensors of
      two general metrics $g$ and $\hat g$ related by the formula 
\begin{equation}\label{eqmetsghatg}
 g = \hat g + \Om^m q
\end{equation}
for a natural number $m \geq 2$, where $q$ is a symmetric tensor and all three tensors $g, \hat g$ and $q$ are at least $C^2$ in a neighbourhood including $\{\Om = 0\}$. No further assumptions besides minimal regularity conditions are imposed on $g$ or $\hat g$,  such as being Einstein or FPG. The result holds therefore in full generality and has potentially a wide range of applications. The proof involves a calculation whose details are given in Appendix \ref{appdomweyl}.

   
   \begin{lemma}\label{lemmaWeyls}
     Let $n \geq 3$ and  $g$, $\hat g$ be $(n+1)$-dimensional  metrics related by \eqref{eqmetsghatg}
     , for $m\geq 2$, with $g, \hat{g}, q$  and $\Om$ at least $C^2$ in a neighbourhood of $\{\Om=0\}$. Assume that $\nabla \Om$ is nowhere null at $\Om=0$. Then their Weyl tensors satisfy the following equation
 \begin{equation}\label{eqexpweyl}
 {C^\mu}_{\nu\alpha\beta } = {\hat C^\mu}_{~~\nu\alpha\beta } -  K_m(\Om) \frac{n-2}{n-1}(\nor^\mu \nor_{[\alpha} \tlt_{\beta]\nu} + {\tlt^\mu}_{~~ [\alpha}\nor_{\beta]}\nor_\nu) + \frac{\epsilon K_m(\Om)}{n-1}({h^\mu}_{[\alpha} \tlt_{\beta]\nu} + \tlt^\mu_{~~ [\alpha} h_{\beta]\nu}) + o(\Om^{m-2})
\end{equation}
with 
\begin{equation}
 K_m(\Om) =  m(m-1) \Om^{m-2} F^2,
\end{equation}
and where $\nabla \Om = F \nor$, for $g(\nor, \nor) = \epsilon = \pm 1$, $h_{\alpha\beta}$ is the projector orthogonal to $\nor$, all indices are raised and lowered with $g$, $t_{\alpha\beta} : = q_{\mu\nu}{h^\mu}_\alpha {h^\nu}_\beta$ while $t$ and $\tlt_{\alpha \beta}$ are its trace and traceless part respectively.
 \end{lemma}

Lemma \ref{lemmaWeyls} has an interesting application in the context of data at $\scri$. Consider a FGP metric $\tilg$ and a geodesic conformal extension $g = \Om^2 \tilg$ and assume that either $n$ is odd or that
     the obstruction tensor is identically zero if $n$ is even.      The FG expansion of this metric allows one to decompose $g = \hat g + \Om^n q$ where $\hat g$ is a metric containing all the terms of the expansion of order strictly lower than $n$ (and possibly also higher
     order terms, but \underline{not} the term at order $n$). The rest of terms are
     collected in $\Om^nq$. By construction all these objects are $C^\infty$ up to and including ${\Omega=0}$ (here we use the assumption that the obstruction tensor vanishes in the even case). Hence all the hypothesis of Lemma \ref{lemmaWeyls} hold with $m=n$. From
   equation \eqref{eqexpweyl}, the $T$-electric part of the Weyl tensors
  of $g$ and of $\hat g$ are related by
\begin{equation}
  (C_T)_{ij}
  = (\hat C_T)_{ij}
  - \Om^{n-2} \lambda^2 n(n-2) \tlt_{ij} + o(\Om^{n-2}), \label{relWeyls}
\end{equation}
It follows immediately from the FG expansion and the definition of $\tlt$ in
Lemma \ref{lemmaWeyls}
that $\tlt_{ij} |_{\Omega=0} = tf(g_{(n)})$ (taking the trace-free part
is unnecessary when $n$ is odd because $g_{(n)}$ is always trace-free in that case).
The tensor $(\hat C_T)_{ij}$
is in general $O(1)$ in $\Omega$, so
$\Omega^{2-n} (\hat C_T)_{ij}$ will generically contain $[(n-1)/2]$ divergent terms,
and the same divergent terms must appear in $\Omega^{2-n} (C_T)_{ij}$ because of
\eqref{relWeyls}. Substracting the divergent terms we get
\begin{equation}
  \label{freedata}
 \left ( \Om^{2-n} (C_T)_{ij}  - \Om^{2-n} (\hat
C_T)_{ij} \right ) |_\scri  = -\lambda^2 n(n-2) tf(g_{(n)}), 
\end{equation}
which provides a general formula for the free data in terms of the
electric parts of the Weyl tensors of  $g$ and $\hat g$ at $\scri$. In the case of $n$ odd more can be said because, as justified below Lemma \ref{lemmaCt}, the regular part of $(\hat C_T)_{ij}$ vanishes at $\scri$
and \eqref{freedata} establishes that $g_{(n)}$ arises as the value of $(C_T)_{ij}$
at $\scri$ once all its divergent terms have been substracted. This last statement is not true in the $n$ even case with zero obstruction tensor, since $\Om^{2-n}(\hat C_T)_{ij}$ may contain regular non-zero terms.

{In the next subsection we will prove that in arbitrary dimension and
for conformally flat $\scri$,  $(\hat C_T)_{ij}$ vanishes so the $T$-electric part of the rescaled Weyl tensor of $g$
actually encodes the trace-free part $tf(g_{(n)})$. In the non-conformally flat case $\Omega^{2-n} C_T$ is generically
    divergent and \eqref{freedata} gives a prescription to remove the divergent terms to retrieve the trace-free part. In the context of AdS/CFT correspondence a
    useful method to remove divergent terms is by means of the so-called      renormalization techniques. One method
    \cite{papasken04,papasken04bis,Sken02} involves decomposing objects in terms of a basis of eigenfunctions of a dilation operator. It would be interesting to analyze whether this method has any relationship with \eqref{freedata}, or whether it can be used to be make the removal of divergent quantities
    more explicit.}


\subsection{Free data and the Weyl tensor}

The aim of this subsection is to determine the role that the electric part of the rescaled Weyl tensor plays in the FG expansion coefficients, with particular interest in the conformally flat $\scri$ case. We  will use formula \eqref{eqCtnocov} to relate the electric part of the rescaled Weyl tensor to the $n$-th order coefficient $g_{(n)}$ of the FG expansion. We start with some preliminary results about umbilical submanifolds (also called {\it totally umbilic}). Recall that a nowhere null submanifold $\Sigma \subset M$ is umbilical if its second fundamental form is 
\begin{equation}
 K_{ij} = f(x^k) \gamma_{ij}
\end{equation}
for a smooth function $f$ of $\Sigma$ and $\gamma$ the induced metric. This property is well-known to be invariant under conformal scalings of the ambient metric.

\begin{lemma}\label{lemmumb1}
 Let $n \geq 2$. Every nowhere null umbilical hypersurface $(\Sigma,\gamma)$ of a conformally flat $(n+1)$-manifold $(M,\hat g)$, where $\gamma$ is induced by $\hat g$, is conformally flat.   
\end{lemma}
\begin{proof}
For $n=2$ the result is immediate as every 2-surface is locally conformally flat, so let us assume $n \geq 3$. Since umbilical submanifolds remain umbilical w.r.t. to the whole conformal class of the metrics and $\hat g$ is conformally flat, then $(\Sigma,\gamma)$ is umbilical w.r.t. the flat metric $g_E = \omega^2 \hat g$. In this gauge, the Gauss equation and its trace by $\gamma$ yield
\begin{align}
R(\gamma)_{ijkl}& = -\epsilon (K_{il} K_{jk} - K_{ik} K_{jl}) = -\epsilon( \gamma_{il} \gamma_{jk} - \gamma_{ik} \gamma_{jl}) \kappa^2,\\
 R{(\gamma)}_{jl} & = -\epsilon (K^2_{jl} - K K_{jl}) = -\epsilon(1 - n) \kappa^2 \gamma_{jl},
\end{align}
where $K_{ij} = \kappa  \gamma_{ij}$ is the second fundamental form, for $\kappa \in \mathbb{R}$ constant as a consequence of the Codazzi equation and the fact that the ambient metric $g_E$ is flat, and $K^2_{ij} := \gamma^{kl}K_{ik} K_{jl}$, $K := \gamma^{ij}K_{ij}$, $\epsilon = \hat g(\nor, \nor)$ with $\nor$ the unit normal to $\Sigma$. The Schouten tensor of $\gamma$ is 
\begin{equation}
 P(\gamma)_{ij} =\frac{1}{n-2}\lr{R(\gamma)_{ij} - \frac{R(\gamma)}{2(n-1)}\gamma_{ij}} =\epsilon \frac{\kappa^2}{2}\gamma_{ij} .
\end{equation}
Thus for $n=3$ we can calculate the Cotton tensor
\begin{equation}
 Y(\gamma)_{ijk} = \nabla_k P(\gamma)_{ij} - \nabla_j P(\gamma)_{ik} = 0,
\end{equation}
and for $n \geq 4$ the Weyl tensor is
\begin{equation}
 C(\gamma)_{ijkl} = R(\gamma)_{ijkl} - \gamma_{ik} P(\gamma)_{jl} + \gamma_{jk} P(\gamma)_{il} + \gamma_{il} P(\gamma)_{jk}-\gamma_{jl} P(\gamma)_{ik} = 0.
\end{equation}
By the standard characterization of locally conformally flat metrics by the vanishing of the Cotton ($n=3$) or Weyl ($n \geq 4$) tensors, the result follows.
\end{proof}

The following results are stated imposing the minimal conditions of differentiability required near $\scri$. We remark than for the cases of our interest, namely FGP metrics, these conditions are always satisfied.

\begin{lemma}\label{lemmumb2}
 Let $g$ and $\hat g$ be metrics related by $g = \hat g + \Om^m q$, where $\Om$ is a defining function of $\Sigma = \{ \Om  = 0\}$ and $g, \hat g$ and $q$ are  $C^1$ in a neighbourhood of $\Sigma$. Then if $m \geq 2$, $\Sigma$ is umbilical w.r.t. $g$ if and only if it is umbilical w.r.t. $\hat g$.
\end{lemma} 
\begin{proof}
 The metrics induced by $g$ and $\hat g$ at $\Sigma$ are the same. Assume that $\Sigma$ is nowhere null. Thus, the property of being umbilical is preserved if the covariant derivatives $\nabla \nor$ and $\hat  \nabla \nor$ w.r.t. the Levi Civita connections of $g$ and $\hat g$ respectively of the normal unit (which is the same for $g$ and $\hat g$) covector $\nor \in (T\Sigma)^\perp$ coincide at $\Sigma$. The inverse metric $g^{-1}$ is $g^{-1} = \hat g^{-1} + \Om^m l$ for $l$ a contravariant tensor $O(1)$ in $\Om$ (cf. equation \eqref{eqmetinvghatg} and argument below). Then, the Christoffel symbols are
 \begin{equation}  
  \Gamma^\mu_{\alpha \beta} = (\hat g^{-1} + \Om^m l)^{\mu \nu}(\partial_{\alpha}(\hat g + \Om^m q)_{\beta\nu} + \partial_{\beta}(\hat g + \Om^m q)_{\alpha\nu}-\partial_{\nu}(\hat g + \Om^m q)_{\alpha\beta}) = \hat \Gamma^\mu_{\alpha \beta} + O(\Om^{m-1}),
 \end{equation}
from which it follows $\nabla \nor \mid_\Sigma = \hat  \nabla \nor \mid_\Sigma$ if $m \geq 2$.  
\end{proof}

Our interest in umbilical submanifolds is because of the (well-known) fact that $\scri$ is umbilical for Poincaré or FGP metrics. This results follows immediately from the Einstein equations at $\scri$, and will be the base for an interesting decompostion that we will derive later in this section (cf. Proposition \ref{propconflatdec}).

\begin{lemma}\label{lemmascriumb}
Let $\tilg$ be a Poincaré or FGP metric for $\scri = (\Sigma,[\gamma])$. Then $\scri$ is umbilical.
\end{lemma}
\begin{proof}
For a geodesic conformal extension $g = \Om^2 \tilg$, the relation between the Ricci tensors of $g$ and $\tilg$ is given by \eqref{eqrelriccis} with $\nabla_\mu \Om \nabla^\mu \Om = -\lambda$ (cf. Lemma \ref{lemmageodesic}). This expression is not defined at $\Om = 0$, but it is when multiplied by $\Om$. Rearranging terms this gives
\begin{equation}\label{eqOmriccgeod}
  \Om R_{\alpha \beta} + (n-1) \nabla_\alpha \nabla_\beta \Om + g_{\alpha \beta} \nabla_\mu \nabla^\mu \Om = \Om(\widetilde{R}_{\alpha\beta} - \lambda n \tilg_{\alpha\beta}),
\end{equation}
where we have used $g = \Om^2 \tilg$ in the RHS. Since $\tilg$ is a Poincaré or FGP metric, the RHS vanishes at $\scri$. This also implies that $g_{\alpha\beta}$ is at least $C^2$ at $\scri$, so $R_{\alpha\beta}$ is defined at $\scri$. In addition writing $\nabla_\alpha \Om = |\lambda|^{1/2}\nor_\alpha$, where $\nor$ is the unit normal of the hypesurfaces $\Sigma_\Om = \{\Om = const. \}$, then $\nabla_i \nabla_j \Om\mid_\scri = |\lambda|^{1/2} K_{ij}$, where $K_{ij}$ is the second fundamental form of $\scri$. Thus, equation \eqref{eqOmriccgeod} gives at $\scri$
\begin{equation}
 (n-1)|\lambda|^{1/2} K_{ij} + f \gamma_{ij} = 0,\quad\quad\mbox{with}\quad\quad f:= \nabla_\mu   \nabla^\mu \Om \mid_{\scri}.
\end{equation} 
\end{proof}

For concreteness, in the remainder of this section, we state and prove  our results in the case of positive cosmological constant and Lorentzian signature. However, they also hold with slight modifications for arbitrary signature and non-vanishing cosmological constant (see Remark \ref{remarkgen} for the specific correspondence).

We start by giving the general form of the FG expansion of the de Sitter spacetime. We refer the reader to \cite{skenderis} for a similar proof in the case of $\lambda<0$. Also, see a discussion of general case in \cite{eastwoodrev} (in terms of Fefferman-Graham ambient metrics).

\begin{lemma}\label{lemmapoinconflat}
 For every Riemmanian conformally flat boundary metric $\gamma$ of dimension $n \geq 3$ and positive cosmological constant $\lambda$, the metric 
 \begin{equation}\label{eqpoindS}
 g = -\frac{\dif \Om^2}{\lambda} + g_{\Om} \quad\quad g_{ij} = 
 {\Ap^k}_i {\Ap^l}_k \gamma_{lj},\quad\quad\mbox{for}\quad {\Ap^k}_i  :=  {\delta^k}_i + \frac{\Om^2}{2\lambda }P_{i l}\gamma^{lk}
\end{equation}
with $P$ the Schouten tensor of $\gamma$, is locally conformally isometric to de Sitter, i.e. $g = \Om^2 \tilg_{dS}$, where $\tilg_{dS}$ is de Sitter.
\end{lemma}

\begin{proof} 
De Sitter spacetime is ACC and its boundary metric $\gamma$ is (by Lemmas \ref{lemmascriumb} and \ref{lemmumb1}) necessarily conformally flat. Moreover, given the freedom in scaling any conformal extension by an arbitrary positive function, any conformally flat metric is (locally) a boundary metric for the de Sitter space. In addition, as a consequence of this fact and Lemma \ref{lemmaexistgeod} we have that for any  conformally flat metric $\gamma$, there exists a local coordinate system of de Sitter near null infinity  such that the metric is in normal form with respect to $\gamma$ (see Remark \ref{remarkPoinnorm}). The core of the proof is to verify that this ACC metric in normal form w.r.t any such conformally flat $\gamma$ takes the explicit form \eqref{eqpoindS}.

Therefore, consider a conformally flat boundary metric $\gamma$ for a geodesic conformal extension of de Sitter $g$. Since de Sitter metric is also conformally flat, it follows that the $T$-electric part of the Weyl tensor $C_T = 0$. Using formula \eqref{eqCtnocov} we obtain the coefficients of the FG expansion, which give the normal form of $g$ w.r.t. $\gamma$. Let us put \eqref{eqCtnocov} in matrix notation
\begin{equation}\label{eqgddotg}
 C_T = \frac{\lambda^2}{2} \left (\frac{1}{2}\dot g_\Om g^{-1}_\Om \dot g_\Om + \frac{1}{\Om} \dot g_\Om - \ddot g_{\Om} \right ) = 0 \quad\Longrightarrow\quad \ddot g_\Om = \frac{1}{2}\dot g_\Om g^{-1}_\Om \dot g_\Om + \frac{1}{\Om} \dot g_\Om
\end{equation}
where a dot denotes derivative w.r.t. $\Om$. First we calculate
\begin{equation}\label{eqdergom}
 \partial_\Om(\dot g_\Om g^{-1}_\Om \dot g_{\Om})= \ddot g_\Om g_\Om^{-1} \dot g_\Om - \dot g_\Om g_\Om^{-1} \dot g_\Om g^{-1}_\Om \dot g_\Om + \dot g_\Om g^{-1}_\Om \ddot g_\Om = \frac{2}{\Om} \dot g_\Om g^{-1}_\Om \dot g_\Om
\end{equation}
where we have used $\partial_\Om {(g_\Om^{-1})} = -g^{-1}_\Om \dot g_\Om g_{\Om}^{-1}$ for the first equality and expression of $\ddot g_\Om$ in \eqref{eqgddotg} for the second equality. Then, taking two derivatives in $\Om$ of  \eqref{eqgddotg} gives
\begin{equation}\label{eqd4gOm}
 \partial^{(4)}_\Om g_\Om =\frac{3}{2 \Om^2} \dot g_\Om g^{-1}_\Om \dot g_\Om.
\end{equation}
Thus, taking one more derivative in $\Om$ of \eqref{eqd4gOm} and combining with \eqref{eqdergom} gives $\partial^{(5)}_\Om g_\Om = 0$ and hence all higher derivates also vanish. Expression \eqref{eqd4gOm} evaluated at $\Om = 0$ gives the expressions for the coefficients (note $\partial^{(k)}_\Om g_\Om\mid_{\Om = 0} = k! g_{(k)}$)
\begin{equation}
  g_{(4)}
  =  \frac{1}{4} g_{(2)} \gamma^{-1}g_{(2)}.
\end{equation}
 The coefficient $g_{(2)}$ can be directly calculated from the recursive relations for the FG expansion and it always coincides with the Schouten tensor of the boundary metric (e.g. \cite{Anderson2004,deHaro}), namely\footnote{{In \cite{Anderson2004} the cosmological parameter $\lambda$ is set to one. The factor $\lambda^{-1}$ arises by simply observing that $\tilg_\lambda = \lambda \tilg$ satisfies the Einstein equations with $\lambda = 1$. Then, the asymptotic analysis of $\tilg_\lambda$ is analogous to the one for $\tilg$ with conformal factor $\Om_\lambda = \lambda^{-1/2}\Om$.}}
 \begin{equation} 
  g_{(2)} = \frac{\lambda^{-1}}{n-2}\lr{Ric(\gamma) - \frac{R(\gamma)}{2(n-1)} \gamma } = P.
 \end{equation}  
 
 Having calculated the only non-zero coefficients $g_{(2)}$ and $g_{(4)}$, it is straightforward to verify that the FG expansion of de Sitter takes the form \eqref{eqpoindS}.
 
 We have shown that for any choice of conformally flat $\gamma$, there exists a de Sitter metric $\tilg_{dS}$  and a choice of conformal factor $\Om$ with associated Gaussian coordinates such that, defining $g$  as in \eqref{eqpoindS}, we have  $\Om^{-2} g = \tilg_{dS}$. Moreover, the metric \eqref{eqpoindS} satisfies all the properties stated in Theorem \ref{theoFG} with the choice $h=0$ if $n\neq 4$ and $h= P^2/(2\lambda)$. Recall that we are assuming $n \geq 3$ and that the obstruction tensor vanishes identically  when $\gamma$ is conformally flat. Now the lemma follows as a consequence of the uniqueness part of the FG expansion stated in Theorem \ref{theoFG}.
\end{proof} 

\begin{remark}\label{remarkgen}
{The result generalizes to arbitrary signature and arbitrary sign of $\lambda$ 
(see \cite{Anderson2004,Sken02} for a discussion on the relation between $\lambda$ positive and negative cases), by changing $\gamma$ to a conformally flat metric of signature $(p,q)$ and $g$ to conformal to a metric of constant curvature (instead of conformal to de Sitter) and signature $(p+1,q)$ if $\lambda>0$ or $(p,q+1)$ if $\lambda<0$. Taking this into account, Proposition \ref{propconflatdec} and Theorem \ref{theognweyl} below easily extend to arbitrary  signature and arbitrary sign of $\lambda$.}
\end{remark}
\begin{remark}\label{remarkpuremag}
 The proof of Lemma \ref{lemmapoinconflat} shows that the condition $C_T = 0$ suffices to obtain a metric of the form \eqref{eqpoindS} with $\gamma$ in an arbitrary conformal class.
 The spacetimes satisfying this condition are the so-called ``purely magnetic'' and they have a long tradition in general relativity (e.g. \cite{barnes04} and references therein). The purely magnetic condition implies restrictive integrability conditions which lead to a conjecture \cite{mcintosh94} that no Einstein spacetimes exist in the $n=3$ case, besides the spaces of constant curvature. Although no general proof has been found so far, the conjecture has been established in restricted cases such as Petrov type $D$, and this not only in dimension four, but in arbitrary dimensions \cite{hervik13}. 
 The explicit form \eqref{eqpoindS} that the metric must take whenever $C_T=0$ gives an avenue to analyze the conjecture in the case of metrics admitting a conformal compactification. This is an interesting problem which, however, falls beyond the scope of the present paper.  
\end{remark}

Before proving Theorem \ref{theognweyl}, we state and prove an auxiliary result (Proposition \ref{propconflatdec}) which is of independent interest since it provides (when combined with Lemma \ref{lemmaWeyls}) a  useful decomposition for calculating leading order terms of the Weyl tensor. This will be exploited in the calculation of  
initial data of spacetimes which admit a smooth conformally flat $\mathscr I$ (cf. Corollary \ref{coroland}).     

 \begin{proposition}\label{propconflatdec}
Assume $n \geq 3$. Let $\tilg$ be a FGP metric with $\lambda$ positive for a Riemannian conformal manifold $\scri = (\Sigma,[\gamma])$. Then $\mathscr I$ is locally conformally flat if and only if any geodesic conformal extension $g = \Omega^2 \tilg$, admits the following decomposition 
  \begin{equation}\label{eqdecconflat}
   g = \hat g + \Om^n q
  \end{equation}
where $\hat g$ is conformally isometric to de Sitter 
and $\hat g$, $q$ and $\Om$ are at least $C^1$ in a neighbourhood of $\{\Om = 0 \}$.
 \end{proposition}
 
 \begin{proof}
 $\scri$ is umbilical w.r.t. $g$ and  if $g$ admits the decomposition \eqref{eqdecconflat}, by Lemma \ref{lemmumb2} $\mathscr I$ is also umbilical w.r.t. $\hat g$. 
 Since $\hat{g}$ is conformally flat, Lemma \ref{lemmumb1} implies that $\scri$ is also conformally flat. This proves the proposition in one direction.
 
  The converse follows by considering the FGP metric in normal form constructed from a representative $\gamma$ in the conformal structure of $\mathscr I$. By assumption, $\gamma$ is conformally flat. The terms up to order $n$ are uniquely generated by $\gamma$ (see Theorem \ref{theoFG} and discussion below). Thus, by Lemma \ref{lemmapoinconflat} 
  \begin{equation}
   g = -\frac{\dif \Om^2}{\lambda} + {\Ap^k}_i {\Ap^l}_k \gamma_{lj} + \Om^n q := \hat g + \Om^n q,
  \end{equation}
  where $\hat g$ is locally conformally isometric to de Sitter and $\hat g$, $q$ and $\Om$ are smooth at $\Om = 0$ by construction.
 \end{proof}

{Now observe that for any set of initial data $(\gamma,g_{(n)})$, one can always add a TT term $\mathring g_{(n)}$ to $g_{(n)}$ so that $(\gamma,g_{(n)} + \mathring g_{(n)})$ gives a new set of initial data. On the other hand, decomposition \eqref{eqdecconflat} in the conformally flat $\scri$ case reads
 \begin{equation}\label{decgn}
   g = -\frac{\dif \Om^2}{\lambda} + {\Ap^k}_i {\Ap^l}_k \gamma_{lj} + \Om^n (\mathring g_{(n)} + \cdots ).
  \end{equation}
  Therefore, if $n \neq 4$ and $n \geq 3$, then $g_{(n)} = \mathring g_{(n)}$ and if $n=4$, then $g_{(4)} = \barg_{(4)} + \mathring g_{(4)}$, where $\barg_{(4)}$ is the term of order four in \eqref{eqpoindS}. This forces $\mathring g_{(n)}$ to be TT, because it is immediate from Lemma \ref{lemmapoinconflat} that de Sitter is given by data $(\gamma, 0)$ for $n \neq 4$ and $n \geq 3$ and by $(\gamma, \barg_{(4)})$ if $n = 4$. Therefore, in the conformally flat $\scri$ case we can always extract the TT term $\mathring g_{(n)}$, which will be called {\bf free part} of $g_{(n)}$.
   Note that a pair $(\gamma,\mathring g_{(n)})$ is equivalent to $(\gamma,g_{(n)})$.}
 
   {We may now extend to the case of arbitrary $\lambda$ the relation
     between the electric part of the rescalled Weyl tensor and the
     coefficient $g_{(n)}$ obtained in \cite{Holl05} for the negative $\lambda$ case. As discussed in the introduction, this extension could be inferred from the general results in \cite{Sken02}. However, our argument is fully conformally covariant and follows directly from the general identity in Lemma \ref{lemmaWeyls}.}

\begin{theorem}\label{theognweyl}
{Assume  $n \geq 3$ and let $\tilg$ be a FGP metric with $\lambda$ positive for a Riemannian conformal manifold $\scri = (\Sigma,[\gamma])$.
 Then, if $\mathscr I$ is conformally flat, $\mathring{g}_{(n)}$, the free part of the $n$-th order coefficient of the FG expansion satisfies
 \begin{equation}
  \Om^{2-n} C_T \mid_{\scri} = -\frac{\lambda^2}{2} n(n-2)  \mathring{g}_{(n)} 
 \end{equation}

 }
\end{theorem}
\begin{proof}
By Proposition \ref{propconflatdec}, admitting a smooth conformally flat $\mathscr I$ amounts to admitting a decomposition of the form \eqref{eqdecconflat}. Then, by Lemma \ref{lemmapoinconflat}, the associated FG expansion has the form 
\begin{equation}
 g = -\frac{\dif\Om^2}{\lambda} + g_\Om = -\frac{\dif\Om^2}{\lambda} + {\Ap^k}_i {\Ap^{l}}_k \gamma_{l j}  + \Om^n \mathring{g}_{(n)} + \cdots = \hat g + \Om^n q ,
\end{equation}
where $ q \mid_{\mathscr I} = \mathring{g}_{(n)}$ and $\hat g$ is  conformal to de Sitter. Using the formula \eqref{eqexpweyl} of Lemma \ref{lemmaWeyls} with $m=n$ and putting $T = |\lambda|^{1/2}u$, with $u$ unit normal, one obtains
\begin{equation}
 (C_T)_{\alpha \beta} = -\frac{\lambda^2}{2} n(n-2)  \tlt_{\alpha \beta} \Om^{n-2} + o(\Om^{n-2})
\end{equation}
and the Theorem follows.

\end{proof}

\begin{remark}
Although this theorem concentrates on the electric part of the Weyl tensor, its proof (which  is based on Lemma \ref{lemmaWeyls}) actually establishes that the full Weyl tensor decays at $\scri$ as $\Omega^{n-2}$. In \cite{OrtaggioWeyl}, the authors analyze the asymptotic behaviour along null godesics of vacuum solutions with non-zero cosmological constant. Letting $r$ be an affine parameter along the geodesics and assuming a priori that suitable components of the Weyl tensor decay at infinity faster than $r^{-2}$ the authors prove a certain peeling behaviour of the Weyl tensor, with the fastest components decaying like $r^{-(n+2)}$
and the slowest as $r^{2-n}$. It is clear that there is a connection between the two results. It would be interesting to establish and analyze this connection, which hopefully would lead to a weakening of the a priori decay rate assumed in \cite{OrtaggioWeyl}.
\end{remark}

\begin{remark}\label{remarpuremag2} 
It is also interesting to comment on the necessary and sufficient conditions for $\mathring{g}_{(n)}$ and $\Om^{2-n}C_T\mid_\scri$ to be the same in the case of Einstein metrics. Just like in the proof of Lemma \ref{lemmapoinconflat}, if $C_T$ has a zero of order $m>3$, we can apply formula \eqref{eqCtnocov} and find
\begin{equation}\label{eqpartial5gom}
 \partial^{(5)}_\Om g_\Om = O(\Om^{m-3})
\end{equation} 
and all coefficients of the FG expansion vanish up to order  $g_{(m+2)}$. If, like in the conformally flat case, $C_T$ has a zero of order $n-2$ , its leading order term determines $\mathring{g}_{(n)}$. 
  If $n$ is odd, we can construct (cf. Theorem \ref{theoanderson})
  two  solutions of the $\Lambda > 0$ Einstein field equations $\hat g$ and $g$
  in a neighbourhood of $\{\Om = 0 \}$, the first one corresponding to the data
  $(\scri, \gamma, 0)$ and the second to the data
  $(\scri, \gamma, g_{(n)})$ where
  $\gamma$ is in an arbitrary conformal class. By the FG expansion we also have
  $g = \hat g + \Om^n q$  with $q = g_{(n)}+O(\Om)$.
  As a consequence of \eqref{eqpartial5gom}, the metric $\hat g$  
  is of the form \eqref{eqpoindS} with $\gamma $ in the given conformal
  class. Then, from equation \eqref{eqCtnocov} it  follows that $\hat g$ is purely magnetic. The converse is also true, namely, if $g = \hat g + \Om^n q$, with $\hat g$ a purely magnetic Einstein spacetime and both $\hat g$, $q$ and $\Om$ are $C^2$ near $\{ \Om = 0 \}$, the electric part of the rescaled Weyl tensor at $\scri$ and $\mathring{g}_{(n)}$ coincide (up to a constant) provided $n>2$. The proof involves simply  taking the $T$-electric part in \eqref{eqexpweyl}. 

This proves that, for Einstein metrics with positive $\Lambda$, of dimension $n+1 \geq 4$ and admitting a conformal compactification, $\mathring{g}_{(n)}$ and $C_T\mid_\scri$ coincide  up to a constant if and only if $g = \hat g + \Om^n q$, where $\hat g$ is a purely magnetic spacetime Einstein with non-zero cosmological constant. However, as mentioned in Remark \ref{remarkpuremag}, it is not clear (and not an easy question) whether purely magnetic Einstein spacetimes are locally isometric to de Sitter or anti-de Sitter spacetime. 
\end{remark}


{
Note that Theorem \ref{theognweyl} has been proven  for metrics of all dimensions $n \geq 3$ and arbitrary signature. An interesting Corollary arises when applying this to the case of $\Lambda>0$ Einstein metrics of Lorentzian signature and odd $n$, because the coefficients of the FG expansion $\gamma $ and $g_{(n)}$ determine initial data at $\scri$ which characterize the spacetime metric (cf. Theorem \ref{theoanderson}). In a similar manner, if $n$ is even and the data $(\gamma,g_{(n)})$ are analytic with $\gamma$ Riemannian, the convergence of the FG expansion (cf. Theorem \ref{theokichen}) it holds in general for any sign of $\Lambda$. Thus, we obtain  a characterization result also for this case.
}

\begin{corollary}\label{coroland}
 {Let $n \geq 3$ be odd. Then for every set asymptotic data $(\Sigma,\gamma,g_{(n)})$ of Einstein's vacuum equations with $\Lambda>0$ and $\gamma$ conformally flat, the free part $\mathring g_{(n)}$ is up to a constant, the electric part of the rescaled Weyl tensor at $\mathscr I$ of the corresponding spacetime. Similarly, if $n \geq 4$ is even, the same statement holds for every analytic asymptotic data $(\Sigma,\gamma,g_{(n)})$, with $\gamma$ Riemannian and for any sign of non-zero $\Lambda$.} 
\end{corollary}

\section{KID for analytic metrics}\label{appendixKID}

In this section we prove a result that determines, in the analytic case, the necessary and sufficient conditions for initial data at $\scri$ so that the corresponding spacetime metric it generates admits a local isometry. The proof relies in the FG expansion of FGP metrics. It is important to remark that analytic metrics correspond to analytic data \cite{ambientmetric}.

This theorem is a generalization to higher dimensions (but restricted to the analytic case) of a known result \cite{KIDPaetz} in dimension $n=3$ establishing that a set $(\Sigma,\gamma, D)$ of asymptotic data at spacelike $\scri$, where $D = g_{(3)}$ (cf. Remark \ref{remarkg3}) is a TT tensor, generates a spacetime admitting one Killing vector field if and only if $g_{(3)}$ satisfies the following Killing Initial Data (KID) equation for $\Y$ a conformal Killing vector field (CKVF) of $\gamma$
 \begin{equation}
 \mathcal{L}_\Y g_{(3)} + \frac{1}{3} \mathrm{div}_\gamma(\Y) g_{(3)}=0.  
\end{equation}
A posteriori $\Y$ it is also the restriction to $\scri$ of the Killing vector of the spacetime whose existence is guaranteed by the theorem.

{In this section we focus in the analytic data case, as we shall require convergence of the FG expansion (cf. Theorem \ref{theokichen}) in the proof of the next theorem.}

\begin{theorem}\label{theoKID}
 Let $\Sigma$ be $n$ dimensional with $n \geq 3$ and let
$(\Sigma,\gamma,g_{(n)})$ be asymptotic data in the analytic class, with $\gamma$ Riemannian and if $n$ even $\obs = 0$. Then, the corresponding spacetime admits a Killing vector field if and only if there exist a CKVF $\Y$ of $\mathscr I$ satisfying the following Killing initial data (KID) equation 
 \begin{equation}\label{eqKID}
 \mathcal{L}_\Y g_{(n)} + \frac{n-2}{n} \mathrm{div}_\gamma(\Y) g_{(n)}=0.
\end{equation}
 \end{theorem}
 \begin{proof}
Showing that \eqref{eqKID} is necessary is proved by direct calculation as follows. Let $X$ be a Killing vector field of $\tilg$ so that
\begin{equation}
 0 = \mathcal{L}_X \tilg = \mathcal{L}_X (\Om^{-2}g) = - 2 \frac{X(\Om)}{\Om^3} g + \frac{1}{\Om^2} \mathcal{L}_X g.
\end{equation}
It follows that on $\mathrm{Int(M)}$, $X$ is a conformal Killing vector of $g$ with a specific right-hand side, namely
\begin{equation}\label{eqconfkill} 
 \mathcal{L}_X g_{\alpha\beta}  = \nabla_\alpha X_\beta + \nabla_\beta X_\alpha =  2 \frac{\mathrm{div}_g X}{n+1} g_{\alpha\beta}, \quad\quad X(\Om) = \frac{\Om}{n+1}\mathrm{div}_g X.
\end{equation}
The following argument  \cite{Fried86initvalue} shows that $X$ must be extendable to $\scri$. The terms $\mathcal{L}_X g_{ 0 \beta}$ of \eqref{eqconfkill} imply a linear, homogeneous symmetric hyperbolic system of propagation equations for $X$. Thus, putting initial data corresponding to $X$ sufficiently close to $\scri$ generates a solution whose domain of dependence must reach $\scri$ (and possibly beyond if the manifold is extendable across $\scri$).
Hence $X$ must admit a smooth extension on $\mathscr I$, which vanishes near $\mathscr I$ only if $X\mid_{\mathscr I} = 0$. The rest of equations $\mathcal{L}_X g_{ ij} $ are also satisfied at $\scri$ by continuity so the extension is a CKVF. 

Then, from the second of equations \eqref{eqconfkill}, it follows that $X(\Om) = 0$ when $\Om = 0$, thus $X$ is tangent to $\mathscr I$, so we denote $\Y :=X\mid_{\mathscr I}$. Putting $g$ in normal form $g = -\frac{\dif \Om^2}{\lambda} + g_\Om$ it easily follows that $\Gamma^\alpha_{\alpha j}  = \Gamma^i_{i j}$. In consequence, expanding $\mathrm{div}_g X$ and evaluating at $\scri$ yields
\begin{align}
 \mathrm{div}_g X\mid_{\mathscr I} & = \partial_\Om(X(\Om))\mid_{\mathscr I} + \partial_j \Y^j + \Gamma^i_{ij}\mid_{\mathscr I} \Y^j \\ 
 & =   \frac{1}{n+1} \mathrm{div}_g X\mid_{\mathscr I}  + ~ \mathrm{div}_\gamma \Y  
  \quad \Longrightarrow\quad 
 \mathrm{div}_g X\mid_{\mathscr I} =  \frac{n+1}{n}\mathrm{div}_\gamma \Y \label{eqdivXdivY}
\end{align}
where we have used the second equation in \eqref{eqconfkill}. 
In addition, the normal form gives the following tangent components of the first equation in \eqref{eqconfkill}:
\begin{equation} 
 \mathcal{L}_X g_\Om =  \frac{2}{n+1} \mathrm{div}_g X g_\Om.
\end{equation}
Evaluating this expression at $\scri$ and taking into account \eqref{eqdivXdivY} shows that $\Y$ is a CKVF of $\gamma$. Also, using the FG expansion of $g_\Om$ we have the following expansion of $\mathcal{L}_X g_\Om$:
\begin{align}
 \mathcal{L}_X g_\Om & = X(\Om) \partial_\Om g_\Om + \mathcal{L}_X \gamma + \Om^2 \mathcal{L}_X g_{(2)}+\cdots +  \Om^n \mathcal{L}_X g_{(n)} + \cdots \\
 & = \frac{\Om}{n+1}(\mathrm{div}_g X) \partial_\Om g_\Om + \mathcal{L}_X \gamma + \Om^2 \mathcal{L}_X g_{(2)}+\cdots +  \Om^n \mathcal{L}_X g_{(n)} + \cdots  .\label{eqliegom}
\end{align}
Therefore
\begin{equation}
\mathcal{L}_X \gamma + \Om^2 \mathcal{L}_X g_{(2)}+\cdots +  \Om^n \mathcal{L}_X g_{(n)} + \cdots  = \frac{1}{n+1} (\mathrm{div}_g X) (2 g_\Om-\Om \partial_\Om g_\Om)\label{eqliegom2}.
\end{equation}
 Equating $n$-th order terms and evaluating at $\scri$ yields \eqref{eqKID} after substituting $\mathrm{div}_g X\mid_\scri$ as in \eqref{eqdivXdivY}.

To prove sufficiency, let us first choose the conformal gauge where $\Y$ is a Killing vector field of $\gamma' = \omega^2 \gamma$. Thus, the corresponding KID equation for $g'_{(n)}$ becomes:
\begin{equation}\label{eqKIDkillgauge}
 \mathcal{L}_\Y g'_{(n)}= 0.
\end{equation}
The remainder of the proof in this gauge, so we drop all the primes. By Lemma \ref{lemmaexistgeod} there exist a geodesic extension which recovers the representative $\gamma$ at $\mathscr I$. In addition, there exists a unique vector field $X$, extended from $\Y$ at $\mathscr I$, which satisfies $[T,X] = 0$. This is obvious in geodesic Gaussian coordinates $\{\Om, x^i\}$, because
\begin{equation}
 [T,X]^\alpha = -\lambda \partial_\Om X^\alpha = 0, 
\end{equation}
with initial conditions $X^\Om \mid_{\Om = 0} = 0$ and $X^i\mid_{\Om = 0} = \Y^i$ has a unique solution $X^\Om = 0$ and $X^i = \Y^i$. We now prove that $X$ is a Killing vector field of the physical metric $\tilg$ provided that \eqref{eqKIDkillgauge} holds. 

Consider the normal form of the metric $g = -\frac{\dif \Om^2}{\lambda} + g_\Om$. Since $\mathcal{L}_{X} \dif \Om = \dif (X(\Om)) = 0$, it follows that $\mathcal{L}_{X} g = \mathcal{L}_{X}(g_\Om)$. Using the FG expansion of $g_\Om$ we have
\begin{equation}
 \mathcal{L}_{X}g_\Om = \mathcal{L}_{X} \gamma + \Omega^2  \mathcal{L}_{X} g_{(2)} + \cdots+ \Omega^n  \mathcal{L}_{X} g_{(n)} + \cdots.
\end{equation}
 If $g$ is analytic, the value of the coefficients $\mathcal{L}_{X} g_{(r)}$ determine $\mathcal{L}_{X} g$ in a neighbourhood of $\scri$. These are
\begin{equation}
 \partial_\Omega^{(r)}  \lr{ \mathcal{L}_{X} g_{\Omega} }\mid_{\Omega=0} = \mathcal{L}_{\Y} \lr{\partial_\Omega^{(r)} g_{\Omega} \mid_{\Omega=0}} = r! \mathcal{L}_{\Y} g_{(r)}.
\end{equation}
We want to show that all these quantities are identically zero, for which we exploit the Feffeman-Graham recursive construction.
The fundamental equation that determines recursively the coefficients of the FG expansion takes the form \cite{Anderson2004}
\begin{equation}\label{eqscoefs}
 - \Om \ddot g_\Om + (n-1) \dot g_\Om - 2 H g_\Om  = \Om L,\quad \quad L := \frac{2}{\lambda} Ric{(g_\Om)}- H \dot g_\Om - (\dot g_\Om)^2 - \frac{2}{\lambda} \widetilde G_\parallel.
\end{equation}
where $\widetilde G_\parallel$ are the tangent components of the tensor $\widetilde G = \widetilde Ric(\tilg) - \lambda n \tilg $ and $\lambda H = \mathrm{Tr}_{g_\Om} A$, with $A := \nabla T$ (cf. Appendix \ref{appA}). A direct calculation shows that taking the $r$-th order derivative of equation \eqref{eqscoefs} at $\Om = 0$ and separating the terms containing highest ($r+1$) order coefficients from the rest gives an expression of the form:
\begin{equation}\label{eqcoefsec}
  (n-r-1)g_{(r+1)} + \lr{\mathrm{Tr}_\gamma g_{(r+1)}}  \gamma = \mathcal{F}_{(r-1)} 
\end{equation}
where $\mathcal{F}_{(r-1)}$ is a sum of terms containing products of coefficients up to order $r-1$ and tangential derivatives thereof, up to second order. Notice that $\widetilde G_\parallel$ vanishes to all orders at $\scri$, so it does not contribute to these equations.
We now prove by induction that the Lie derivative of all cofficients vanish provided equation \eqref{eqKIDkillgauge} is satisfied. 

First, the Lie derivative of \eqref{eqcoefsec}, given that $\Y$ is a Killing of $\gamma$, yields
\begin{equation}\label{eqcoefseclie}
  (n-r-1)\mathcal{L}_\Y g_{(r+1)} + \lr{\mathrm{Tr}_\gamma \mathcal{L}_\Y g_{(r+1)}}  \gamma = \mathcal{L}_\Y \mathcal{F}_{(r-1)} .
\end{equation}
Assume by hypothesis that the Lie derivative $\mathcal{L}_\Y$ of all the coefficients up to a certain order $r$ is zero (for the moment we do not assume neither $r<n$ nor $r>n$).  
The Lie derivative $\mathcal{L}_{\Y}\mathcal{F}^{(r-1)}$ is a sum where each terms is multiplied by either $\mathcal{L}_{\Y}g_{(s)}$, $\mathcal{L}_{\Y}\partial_i g_{(s)}$ or $\mathcal{L}_{\Y}\partial_i \partial_j g_{(s)}$ , with $s \leq r-1$. Since $\Y$ commutes with $T = - \lambda \partial_\Om$, we can adapt coordinates to both vector fields, namely $\Y = \partial_j$, so that in these coordinates $\mathcal{L}_{\Y}\partial_i g_{(s)} = \partial_i \mathcal{L}_{\Y} g_{(s)}$ and  $\mathcal{L}_{\Y}\partial_i \partial_j g_{(s)} = \partial_i \partial_j \mathcal{L}_{\Y} g_{(s)}$. Thus each term in  $\mathcal{L}_{\Y}\mathcal{F}^{(r-1)}$ contains a Lie derivative $\mathcal{L}_{\Y}g_{(s)}$ with $s<r-1$, or a tangential derivative thereof up to second order. Thus by the induction hypothesis $\mathcal{L}_{\Y}\mathcal{F}^{(r-1)} = 0$. Therefore, it follows that $\mathcal{L}_\Y g_{(r+1)} = 0$

The induction hypothesis can be assumed for $r<n-1$ because it is true for the first term $\mathcal{L}_\Y \gamma = 0$ and we have equations for the succesive terms. For $r = n-1$ the fundamental equation does not determine the term $g_{(n)}$ any longer (this is the reason why this terms is free-data in the FG expansion), so the induction hypothesis cannot go further in principle. But since we are imposing the condition $\mathcal{L}_\Y g_{(n)} = 0$, the induction hypothesis can be extended to any value of $r$. Therefore, all the derivatives $\mathcal{L}_\Y g_{(r+1)}$ vanish, so if $g$ is analytic $\mathcal{L}_\Y g = 0$. 
\end{proof}

 In short, the argument behind the proof of Theorem \ref{theoKID} relies on the well-known fact that the recursive relations that determine
   the coefficients
   of the FG expansion can be cast in a covariant form, so that ultimately
   all terms can be expressed in terms of $\gamma$, its curvature tensor,
   $g_{(n)}$ and covariant derivatives thereof.
   Then the Lie derivative of any coefficient must be zero provided that $\mathcal{L}_\Y \gamma = \mathcal{L}_\Y g_{(n)} = 0$. The case with non-zero obstruction tensor, and hence involving logarithmic terms is likely to
   admit an analogous proof. However, the recursive equations equivalent to \eqref{eqcoefsec} are not so explicit, because taking derivatives of order higher than $n$ yields an expression which mixes up coefficients of the regular part $g_{(r)}$ and logarithmic terms $\obs_{(r,s)}$ of the expansion. These expressions are notably more involved (see e.g. \cite{Rendall03}). If one showed
   that every coefficient $\obs_{(r,s)}$ admits a covariant form which only involves geometric objects constructed from  $\gamma$, $g_{(n)}$ and its covariant derivatives, a similar argument as in the proof above would establish
   that equation \eqref{eqKID} is also sufficient for the
  spacetime to admit a  Killing vector field in the case of 
  analytic data with non-vanishing $\obs$.  It is hard to imagine that this is not the case, and in fact the result should follow  from the expressions in \cite{Rendall03}, but the details need to be worked out. On the other hand, the necessity of \eqref{eqKID} is true  in general and the argument
 is totally analogous to the one presented above  except that equations \eqref{eqliegom} and \eqref{eqliegom2} contain also logarithmic terms. 
 We will not discuss this case any further since for the rest of paper we shall focus on conformally flat $\scri$ (hence $\obs = 0$). We plan to come back to this open issue in a future work.

\section{Characterization of generalized Kerr-de Sitter metrics}\label{secKdsmet}

In this section, we will apply the results obtained in the previous sections to find a characterization of the higher dimensional Kerr-de Sitter metrics. These were firstly formulated in five dimensions in \cite{Hawk99} and latter extended {to arbitrary dimensions} in \cite{Gibbons2005}. We first prove that these metrics admit a smooth conformally flat $\scri$. Then we combine with Theorem \ref{theognweyl} to determine their initial data at $\scri$, which is straightforwardly computable from equation \eqref{eqexpweyl}.
The data corresponding to Kerr-de Sitter in all dimensions  are analytic. Therefore, by Theorem \ref{theokichen}, the identification of their data provides a characterization of the metric also in the case of $n$ even. Hence, we perform the analysis simultaneously  for $n$ even and odd. We stress that the initial data for Kerr-de Sitter in any dimensions turn out to be a natural extension of the data in the $n=3$ case \cite{KdSlike},\cite{KdSnullinfty}, which in turn is directly related to a local spacetime characterization of the Kerr-NUT-de Sitter family of metrics \cite{marsseno15}. This 
result provides a clear geometric reason why the higher-dimensional Kerr-de Sitter metrics in \cite{Hawk99}, \cite{Gibbons2005} are indeed directly linked to the Kerr-de Sitter metric in spacetime dimension four.

Like in the four dimensional case, the generalized Kerr-de Sitter metrics are $(n+1)$-dimensional Kerr-Schild type metrics. Namely, they admit the following form 
\begin{equation}
 \tilg = \tilg_{dS} + \widetilde{\mathcal{H}}~\widetilde k\otimes \widetilde k
\end{equation}
with $\widetilde{g}_{dS}$ the de Sitter metric, $k$ is a null (w.r.t. to both $\tilg$ and $\widetilde{g}_{dS}$) field of 1-forms and $\widetilde{\mathcal{H}}$ is a smooth function. In order to unify the $n$ odd and $n$ even cases in one single expression, we define the following parameters
\begin{equation}
 p := \lrbrkt{\frac{n+1}{2}} -1,\quad\quad q := \lrbrkt{\frac{n}{2}},
\end{equation}
where note, $p= q$ if $n$ odd and $p + 1= q$ if n even. The explicit expression of the Kerr-de Sitter metrics will be given using the so-called ``spheroidal coordinates'' $\{r,\m_i\}_{i=1}^{p+1}$ (see \cite{Gibbons2005} for their detailed construction), with the redefinition $\rho  := r^{-1}$.  Strictly speaking, they do not quite define a  coordinate system because the $\m_i$ functions are constrainted to satisfy
\begin{equation}\label{DEFmu}
  \sum\limits_{i=1}^{p+1}  \m_i^2 = 1.
\end{equation}
However, it is safe to abuse the language and still call $\{\m_i \}$  coordinates. To complete $\{\rho, \m_i \}$  to full spacetime coordinates we include $\{\rho, t, \{\m_i\}_{i=1}^{p+1}, \{\phi_i\}_{i=1}^q\}$. The $\m_i$s and $\phi_i$s are related to polar and azimuthal angles of the sphere respectively and they take values in $0 \leq \m_i \leq 1$ and $0 \leq \phi_i < 2 \pi$ for $i = 1, \cdots, q$ and (only when $n$ odd) $-1 \leq \m_{p+1} \leq 1$. Associated to each $\phi_i$ there is one rotation parameter $a_i \in \mathbb{R}$. For notational reasons, it is useful to define a trivial parameter
$a_{p+1} =0$ in the case of $n$ odd. The remaining $\rho$ and $t$  lie in $0 \leq \rho < \lambda^{1/2}$ and $t \in \mathbb{R}$. The domain of definition of $\rho$ can be extended (across the Killing horizon) to $\rho > \lambda^{1/2}$, but this is unnecessary in this work since we are interested in regions near $\rho = 0$.

In addition, as we will work with the conformally extended metric $g = \rho^2 \tilg$, we directly write down the expresions of the following quantities, which admit a smooth extension to  $\rho=0$,
\begin{equation}\label{eqredefconf}
 \hat g = \rho^2 \tilg_{dS},\quad\quad \mathcal{H} = \rho^2 \widetilde{\mathcal{H}},\quad\quad k_\alpha = \widetilde k_\alpha
\end{equation}
and
\begin{equation}\label{eqgKS}
 g = \hat g + \mathcal{H}~ k \otimes k.
\end{equation}
We provide below the expression of $k_{\alpha}$ (as opposed to $k^{\alpha}$) because the metrically associated vector field $k^\alpha = g^{\alpha\beta} k_\beta$ is no longer the same as $\widetilde k^\alpha = \tilg^{\alpha\beta} \widetilde k_{\beta}$. In order for the reader to compare with the original publication \cite{Gibbons2005}, we remark that the expressions given there are for the ``physical'' objects $\tilg_{dS},~\widetilde{\mathcal{H}}, \widetilde k$, using the coordinates $r := \rho^{-1}$ and denoting $\mu_i := \m_i$ instead.

Let us now introduce the functions
\begin{equation}\label{eqdefWXi}
 W := \sum\limits_{i=1}^{p+1} \frac{\m_i^2}{1 + \lambda a_i^2}\quad \quad  \Xi := \sum\limits_{i=1}^{p+1} \frac{\m_i^2}{1 + \rho^2 a_i^2},
\end{equation}
Note that it is thanks to having introduced the spurious quantity $a_{p+1} \equiv 0$ that these expressions take a unified form in the $n$ odd and $n$ even cases. The explicit form of the objects in \eqref{eqredefconf} in the case of generalized Kerr-de Sitter are     
\begin{eqnarray}
  \hat{g} & = & - W (\rho^2 - \lambda ) \dif t^2 + \frac{\Xi}{\rho^2 - \lambda } \dif \rho^2 + \delta_{p,q} \dif \m_{p+1}^2 + \sum_{i=1}^q \frac{1 + \rho^2 a_i^2}{1 + \lambda a_i^2} \lr{\dif \m_i^2 + \m_i^2 \dif \phi_i^2} 
  \\
 & & 
  + \frac{\lambda}{W (\rho^2- \lambda )} \lr{\sum_{i= 1}^{p+1} \frac{\lr{1 + \rho^2 a_i^2} \m_i \dif \m_i}{1 + \lambda a_i^2}}^2 , \label{eqgdSodd}\\
 k & = & W \dif t - \frac{\Xi}{\rho^2 - \lambda } \dif \rho - \sum\limits_{i=1}^q \frac{a_i \m_i^2}{1 + \lambda a_i^2} \dif \phi_i, \label{eqkodd}\\
  \Pi & = & \prod\limits_{j=1}^q (1 + \rho^2 a_j^2), \quad \mathcal{H} =  \frac{2M \rho^{n}}{ \Pi \Xi}, \quad M\in \mathbb{R}  \label{eqgHodd}.
\end{eqnarray}
The term $\delta_{p,q}$ only appears when $q = p$, i.e. when $n$ is odd. In the case of even $n$, all terms multiplying $\delta_{p,q}$ simply go away.

The function $\mathcal{H}=O(\rho^n)$ and $k\otimes k = O(1)$. Therefore $g$ decomposes as 
\begin{equation}
 g = \hat g + \rho^n q,\quad\quad\mbox{with}\quad\quad q = \frac{\mathcal{H}}{\rho^n} k \otimes k = O(1).
 \end{equation}
 Let $\gamma$ be the metric induced at $\Sigma = \{ \rho = 0 \}$ by $g$. By Lemma \ref{lemmaexistgeod}, we can define a geodesic conformal factor $\Om$  such that $\{ \Om = 0 \} = \Sigma$ and which induces the same metric $\gamma$ at $\Sigma$. Hence $\Om = O(\rho)$ and therefore $\mathcal{H} = O(\Om^n) $ and $q = O(1)$ (in $\Om$). So by Proposition \ref{propconflatdec} it follows that the generalized Kerr-de Sitter metrics in all dimensions admit a conformally flat $\mathscr I$. 
This can be also verified by explicit calculation. From \eqref{eqgdSodd}, the induced metric at $\{ \rho = 0 \}$ has the following expression
\begin{align}
   \gamma & =   \lambda W \dif t^2  + \delta_{p,q} \dif \m_{p+1}^2 + \sum_{i=1}^{q} \frac{\dif \m_i^2 + \m_i^2 \dif \phi_i^2}{1 + \lambda a_i^2} 
  - \frac{1}{W} \lr{\sum_{i= 1}^{p+1} \frac{ \m_i \dif \m_i}{1 + \lambda a_i^2}}^2.\label{eqgammu}
\end{align}
It is useful to define new coordinates
\begin{equation}\label{eqhatmuscr}
 \hat \m_i^2 : = \frac{1}{W} \frac{\m_i^2}{1 + \lambda a_i^2},
\end{equation}
which from \eqref{eqdefWXi} are restricted to satisfy $\sum_{i=1}^{p+1}{\hat {\m}}_i^2 =1$. Since also $\sum_{i=1}^{p+1} \m_i^2 =1$, this allows us to express $W$ (given in \eqref{eqdefWXi}) in terms of the hatted coordinates
\begin{equation}\label{eqhatW} 
 W = \frac{1}{1+\sum_{i=1}^{p+1}{\lambda \hat{\m}}_i^2 a_i^2}.
\end{equation}
A direct calculation shows that the metric \eqref{eqgammu} expressed with $\hat \m_i$s takes the form 
\begin{equation}\label{eqgammhatmu}
 \gamma = W \Big( \lambda \dif t^2 + \delta_{p,q}\dif \hat \m_{p+1}^2 +  \sum\limits_{i = 1}^{q}\big(\dif \hat \m_i^2 + \hat \m_i^2  \dif \phi_i^2\big)\Big)\mid_{\sum_{i=1}^{p+1}{\hat{\m}}_i^2 =1}.
\end{equation}
An explicitly flat representative of the conformal class of $\gamma$ can be obtained using the coordinates
\begin{equation}\label{eqcartcoor}
 x_i := e^{\sqrt{\lambda}t} \hat \m_i \cos \phi_i \quad \quad y_i := e^{\sqrt{\lambda}t} \hat \m_i \sin  \phi_i,\quad i= 1,\cdots , q
\end{equation}
together with $z := e^{\sqrt{\lambda}t} \hat \m_{p+1}$ if $n$ odd, which are Cartesian for the following flat metric
\begin{equation}\label{eqdefgammE}
 \gamma_E := \frac{e^{2\sqrt{\lambda}t}}{W} \gamma = \delta_{p,q} \dif z^2 + \sum\limits_{i = 1}^{q}\big(\dif x_i^2 + \dif y_i^2\big).
\end{equation}
This form will be used below to determine the conformal class of a conformal Killing vector $\Y$ which we introduce next. Let us denote the projection of $k$ onto $\mathscr I$ by
\begin{equation}
 \Y_\alpha = \big(k_\alpha + (k_\beta \nor^\beta) \nor_\alpha\big)\mid_\scri
\end{equation}
with $\nor_\alpha = \nabla_\alpha \rho/|\nabla \rho|_g$ the unit timelike normal to $\mathscr I$. Explicitly
\begin{align}\label{eqYflat}
 \Yf = W \dif t - \sum\limits_{i=1}^{q} \frac{a_i \m_i^2}{1 + \lambda a_i^2} \dif \phi_i = W\Big(\dif t - \sum\limits_{i=1}^{q} \hat \m_i^2 a_i \dif \phi_i \Big),
\end{align}
where from now on the metrically related one-form $\Yf = \gamma(\Yv,\cdot)$ associated to a CKVF $\Yv$ of $ \gamma$ will be written in boldface when index free notation is used. The vector field is, using \eqref{eqgammhatmu},
 \begin{equation}\label{eqYsharp}
  \Yv =\frac{1}{\lambda} \partial_t - \sum^{q}_{i=1} a_i \partial_{\phi_i},
 \end{equation} 
 which in Cartesian coordinates \eqref{eqcartcoor} of $\gamma$ takes the form 
  \begin{equation}\label{eqYsharpcart1}
  \Yv =\frac{1}{\sqrt{\lambda}}\widetilde \Yv   - \sum^{q}_{i=1} a_i \eta_i 
 \end{equation} 
where we have introduced
\begin{equation}\label{eqYsharpcart2}
 \widetilde \Yv := \delta_{p,q} \partial_z+ \sum^{q}_{i=1} x_i \partial_{x_i} + y_i \partial_{y_i} 
 ,\quad\quad
 \eta_i  :=  x_i \partial_{y_i} - y_i \partial_{x_i}.
\end{equation}
The vector $\widetilde \Yv$ is a homothety of $\gamma_E$ and
each $\eta_i$ is a rotation of this metric. Consequently, $\Yv$ is a CKVF of $\gamma$.

 The electric part of the rescaled Weyl tensor at $\scri$ can be obtained at once from expression \eqref{CtCperp} and Lemma \ref{lemmaWeyls} using $\Omega = \rho$ and $m = n$, because by definition $(t_{\alpha \beta})\mid_{\scri} = (\mathcal H/\rho^n)\mid_{\scri} \xi_{\alpha} \xi_{\beta}$
and $ \tlt\mid_{\scri}$ is its trace-free part. Note also that $H/\rho^n\mid_{\scri} = 2M$. Thus
\begin{equation}\label{eqDKds1}
 D =  \rho^{2-n}{C^\mu}_{\alpha \nu\beta}\nabla_\mu \rho \nabla^\nu \rho\mid_{\mathscr{I}} = -\frac{1}{2} \lambda n (n-2) \tlt_{\alpha \beta}\mid_\scri =  -M \lambda n (n-2) \lr{\Y_\alpha \Y_\beta - \frac{|\Y|_\gamma^2}{n} \gamma_{\alpha \beta}}. 
\end{equation}
Since, by equation \eqref{eqhatW} above, 
  \begin{equation}
   |\Y|_\gamma^2 = W \lr{\frac{1}{\lambda} + \sum_{i=1}^q a_i^2 \hat \m_i^2} = \frac{1}{\lambda},
  \end{equation}
 $D$  can be cast as
\begin{equation}
 D = \con D_\xi,\quad\quad \mbox{with}\quad \con := -\frac{M n(n-2)}{\lambda^{\frac{n}{2}}} 
\end{equation}
and
\begin{equation}\label{eqDKds}
 D_\Y :=   \frac{1}{|\Y|_\gamma^{n+2}}\lr{\Yf \otimes \Yf - \frac{|\Y|_\gamma^2}{n} \gamma}.
\end{equation}

\begin{remark}\label{remarknotationD}
 Following the notation in \cite{KdSlike}, observe that the primary object defining
$D_{\Yv}$ is the vector field $\Yv$, while in the RHS of \eqref{eqDKds} there appears the one-form
$\Yf = \gamma(\Yv,\cdot)$.
This notation generalizes to any CKVF $\Yv'$ and metric $\gamma'$.
This
will be useful in order to prove conformal properties of $D_\Yv$ which depend only on $\Yv$ (cf.
Lemma \ref{lemmaequivdata} below).
\end{remark}

Summarizing, we have proven the following result.

\begin{proposition}
 The asymptotic data  corresponding to the $(n+1)$-dimensional generalized Kerr-de Sitter metrics is given by the class of conformally flat metrics and the class of TT tensors determined by \eqref{eqDKds}, where $\Y$ is the field of 1-forms given by \eqref{eqYflat} when the metric $\gamma$ is written in the coordinates where \eqref{eqgammhatmu} holds.
\end{proposition}

Now suppose that we let $\Y$ to be any CKVF of $\gamma$. By direct calculation one shows that the corresponding $D_\Y$ is still TT w.r.t. $\gamma$. The spacetimes corresponding to the class of data obtained in this way constitute a natural extension to arbitrary dimensions of the so-called {\it Kerr-de Sitter-like class} with conformally flat $\scri$, first defined for $n=3$ in \cite{KdSlike} and \cite{KdSnullinfty}. A detailed definition and properties of this class of spacetimes will be discussed in   \cite{marspeonKSKdS21}. What we want to prove now is that for data $(\Sigma, \gamma, D_\Y)$ with $\gamma$ conformally flat and $\Y$ a CKVF of $\gamma$, only the conformal class of $\Y$ (equivalently $\Yv$) matters. Then, by identifying the conformal class of \eqref{eqYsharp} we will obtain a complete geometrical characterization of Kerr-de Sitter in all dimensions. 
 
Recall that the conformal class of a CKVF $\Yv$ of a manifold $(\Sigma,\gamma)$ is given by the pushforward vector fields $\phi_\star(\Y)$ for every local conformal differmorphism  $\phi$ of $(\Sigma,\gamma)$. Such diffeormorphisms are only locally defined in a neighbourhood $\mathcal{U} \subset \Sigma$, i.e. $\phi: \mathcal{U} \rightarrow \Sigma \mid \phi^\star(\gamma) = \omega^2 \gamma$ for some smooth positive function $\omega$ of $\mathcal U$. Thus, the  equivalence relation between CKVFs makes sense when we restrict to $\mathcal{U} \cap \phi(\mathcal{U})$. For example, this is the case for M\"obius transformations in the Euclidean space $\mathbb{E}^n$ (e.g. \cite{IntroCFTschBook,blairconformal}) and generalizes to all locally conformally flat manifolds. It is however well-known (e.g. \cite{IntroCFTschBook}) that the  M\"obius tranformations are global in the $n$-sphere, so one may use the conformal classes globally defined in the sphere as a reference for all locally conformally flat manifolds. Details of this construction are given in \cite{marspeonKSKdS21}. In the following, $\mathrm{ConfLoc}(\Sigma,\gamma)$ denotes the set of all local conformal diffeomorphisms and when referring to conformal classes of CKVFs we shall implicitly restrict to the domain where this relation makes sense.

\begin{lemma}\label{lemmaequivdata}
 For asymptotic data $(\Sigma, \gamma, \kappa D_\Yv)$ with $\gamma$ conformally flat and any
trasformation   $\phi \in \mathrm{ConfLoc}(\Sigma, \gamma)$, the following equivalence of data holds
\begin{equation}\label{eqequivdata}
 (\Sigma, \gamma, \kappa D_{\phi_\star \Yv}) \simeq (\Sigma, \phi^\star \gamma, \phi^\star(\kappa D_{\phi_\star \Yv})) =  
 (\Sigma, \omega^2 \gamma,\omega^{2-n} \kappa D_\Yv) \simeq (\Sigma, \gamma,\kappa D_\Yv),
\end{equation}
 where the tensor $D_{\phi_\star \Yv}$ is given by \eqref{eqDKds} with the notation of Remark \ref{remarknotationD}.
\end{lemma}
\begin{proof}
 The first equivalence in \eqref{eqequivdata} is a consequence of the diffeomorphism equivalence of data and the last one is a consequence of Corollary \ref{coroland} and the conformal properties of the Weyl tensor. Hence
 we must verify the equality in the expression. Denote the one-form $\phi_\star(\Yf) : = \gamma(\phi_\star \Yv, \cdot )$. For any vector field $X\in T\Sigma$ we compute
\begin{equation}
 (\phi^\star \phi_\star(\Yf)) (X) = (\phi_\star(\Yf)) (\phi_\star X) = \gamma(\phi_\star \Yv, \phi_\star X) = \omega^2 \gamma(\Yv, X) = \omega^2 \Yf(X),
\end{equation}
which means $\phi^\star(\phi_\star(\Yf)) = \omega^2 \Yf$. Moreover  $|\phi_\star(\Y)|_\gamma =\sqrt{\gamma(\phi_\star \Yv, \phi_\star \Yv)} = \omega |\Y|_\gamma$.
Thus
\begin{align}
 \phi^\star \lr{D_{\phi_\star(\Yv)}} & =   \frac{1}{|\phi_\star(\Y)|^{n+2}_\gamma} \lr{
\phi^\star(\phi_\star(\Yf) \otimes \phi_\star(\Yf)) - \frac{|\phi_\star(\Y)|_\gamma^2}{n} \phi^\star \gamma} \\ &  =  \omega^{-n+2} \frac{1}{|\Y|^{n+2}_\gamma} \lr{
\Yf \otimes \Yf - \frac{|\Y|_\gamma^2}{n} \gamma} = \omega^{2-n} D_\Yv.
\end{align}
\end{proof}

We now come back to Kerr-de Sitter and identify the conformal class of \eqref{eqYsharp}. Following the results in \cite{marspeon21}, a direct way to do that is to write $\Yv$ in any Cartesian coordinate system for any flat representative $\gamma_E$ in the conformal class of metrics. One then associates to the explicit form of $\Yv$ in these coordinates a skew-symmetric endomorphism (equivalently a two-form) of  the Minkowski spacetime $\mathbb{M}^{1,n+1}$. Let us denote this set $\skwend{\mathbb{M}^{1,n+1}}$ and we write $F(\Yv)$ for the skew-symmetric endomorphism associated to the CKVF $\Yv$. The classification of 
$\skwend{\mathbb{M}^{1,n+1}}$ up to $O(1,n+1)$ (i.e. the Lorentz group) transformations is equivalent to the classification of CKVFs up to conformal transformations, but the former is simpler because $\skwend{\mathbb{M}^{1,n+1}}$ are linear operators, while the CKVFs are vector fields (see \cite{marspeon21} for additional details).

Consider Cartesian coordinates $\{X^\alpha\}_{A=1}^n$ for an $n$-dimensional Riemannian flat metric. Then an arbitrary CKVF is well-known (e.g. \cite{IntroCFTschBook}) to be given by  
\begin{equation}  \label{eqCKVF}
\Yv  = \big(\br^A + \nu X^A + (\ar_B X^B)X^A - \frac{1}{2}(X_B X^B)\ar^A - {\omega^A}_B X^B \big) \partial_{X^A},
\end{equation}
to which one associates \cite{KdSlike},\cite{marspeon21} the following skew-symmetric endomorphism, given in an orthonormal basis $\{e_\alpha \}_{\alpha = 0}^{n+1}$ of $\mathbb{M}^{1,n+1}$ with $e_0$ timelike,
\begin{equation}\label{skwmatrix}
 F(\Yv) = \begin{pmatrix}
  0 &  -\nu & -\ar^t + \br^t/2 \\
  -\nu & 0 & - \ar^t - \br^t/2 \\
  -\ar + \br/2 & \ar +\br/2 & -\pmb{\omega}
  \end{pmatrix}.
\end{equation}
where $\ar,\br \in \mathbb{R}^n$ are column vectors with components $\ar^A, \br^A$ respectively, $t$ stands for their transpose (row vector), $\nu \in \mathbb{R}$ and ${\bm \omega}$ is $n\times n$ real skew-symmetric matrix of components $( \delta_{AC} {\omega^C}_A =:){\omega}_{AB} = -\omega_{BA}$. 
 Understood as a map $F: \Yv \mapsto F(\Yv)$, $F$ is a Lie algebra anti-homomorphism, i.e. 
$[F(\Yv),F({\Y}')] = -F([\Yv,{\Y}'])$. The $O(1,n+1)$ transformations on $F(\Yv)$ are translated into conformal tranformations of $\Yv$. That is, for every $\Lambda \in O(1,n+1)$, then $\Lambda \cdot F(\Yv) = F (\phi_{\Lambda \star} (\Yv))$ for a conformal transformation $\phi_\Lambda$ of $\gamma_E$, where ``dot'' denotes adjoint action, which in matrix notation corresponds simply to the multiplication of matrices $\Lambda \cdot F(\Yv) = \Lambda F(\Yv) \Lambda^{-1}$. As a consequence, the classification of 
$\skwend{\mathbb{M}^{1,n+1}}$ up to $O(1,n+1)$ transformations is equivalent to the classification of CKVFs up to conformal transformations.

In \cite{marspeon21}, it is proven that the orbits of $\skwend{\mathbb{M}^{1,n+1}}/O(1,n+1)$ are characterized\footnote{Equivalent characterizations may be found in the literature (e.g. \cite{KdSlike}) in terms of traces of even powers of $F(\Yv)$ and matrix rank.} by the eigenvalues of $-F(\Yv)^2$ together with the causal character of $\ker F{(\Yv)}$. The algorithm is as follows. Denote by $\mathcal{P}_{F^2}(-x)$ the characteristic polynomial of $-F(\Yv)^2$ and define
\begin{equation}\label{defQF2}
\mathcal{Q}_{F^2}(x) := \lr{\mathcal{P}_{F^2}(-x)}^{1/2} \quad \mbox{($n$ even)}, \quad\quad \mathcal{Q}_{F^2}(x) := \lr{\frac{\mathcal{P}_{F^2}(-x)}{x}}^{1/2} \quad \mbox{($n$ odd)}.
\end{equation}
From the properties of $F^2(\Yv)$, it follows \cite{marspeon21} that 
$\mathcal{Q}_{F^2}(x)$ is a polynomial of degree $q+1$ with $q+1$ real roots counting multiplicity, with at most one of which negative. Then   
\begin{proposition}\label{defgammamu}
 Let $\mathrm{Roots}\lr{\mathcal{Q}_{F^2}}$ denote the set of roots of $\mathcal{Q}_{F^2}(x)$ repeated as many times as their multiplicity and sorted as follows
 \begin{enumerate}
  \item[a)] If $n$ odd, $\lrbrace{\sigma; \mu_1^2, \cdots, \mu_{p}^2} := \mathrm{Roots}\lr{\mathcal{Q}_{F^2}}$ sorted by $\sigma \geq \mu_1^2\geq \cdots \geq \mu_{p}^2$ if $\ker F(\Yv)$ is timelike and $\mu_1^2\geq \cdots \geq \mu_{p}^2\geq 0 \geq \sigma$ otherwise. 
  \item[b)] If $n$ even, 
  $\lrbrace{-\mu_t^2, \mu_s^2; \mu_1^2, \cdots, \mu_{p}^2} := \mathrm{Roots}\lr{\mathcal{Q}_{F^2}}$ sorted by $\mu_1^2\geq \cdots \geq \mu_{p}^2\geq\mu_s^2 = -\mu_t^2 = 0$ if $\ker F(\Yv)$ is degenerate and $\mu_s^2 \geq \mu_1^2\geq \cdots \geq \mu_{p}^2\geq 0 \geq  -\mu_t^2$ otherwise.  
 \end{enumerate}
 Then the parameters $\lrbrace{\sigma; \mu_1^2, \cdots, \mu_{p}^2}$ for $n$ odd and $\lrbrace{-\mu_t^2, \mu_s^2; \mu_1^2, \cdots, \mu_{p}^2}$  for $n$ even determine uniquely the class of $F(\Yv)$ up to $O(1,n+1)$ transformations and hence also the class of $\Yv$ up to conformal transformations.  
\end{proposition}

We now apply these results to the Kerr-de Sitter case.
We have already obtained a flat representative $\gamma_E$ and have introduced corresponding Cartesian coordinates \eqref{eqdefgammE} . We have also obtained the explicit form of $\Yv$ in these coordinates, namely \eqref{eqYsharpcart1} and \eqref{eqYsharpcart2}. Then, from Proposition \ref{defgammamu} it is straightforward to identify the conformal class of $\Yv$.
Denote the Cartesian coordinates in \eqref{eqcartcoor} by $\{X\}_{A = 1}^{n} = \{z, \{x_i,y_i \}_{i=1}^q\}$ if $n$ odd and $\{X\}_{A = 1}^{n} = \{x_i,y_i \}_{i=1}^q$ if $n$ even.
From equations \eqref{eqYsharpcart1}, \eqref{eqYsharpcart2} the parameters of $\Yv$ written as in \eqref{eqCKVF} are $\nu = \lambda^{-1/2},~\ar^A = \br^A = 0$ and $\omega_{AB} = 2 a_i {\delta^{2i}}_{[A}{\delta^{2i+1}}_{B]}$ for $n$ odd and  $\omega_{AB} = 2 a_i {\delta^{2i-1}}_{[A}{\delta^{2i}}_{B]}$ for $n$ even. Thus, from equation \eqref{skwmatrix} it is immediate
\begin{align} 
 F(\Yv) & = \left(
 \begin{array}{cc}
  0 & -\lambda^{-1/2} \\ -\lambda^{-1/2} & 0
 \end{array}
 \right) \oplus (0)
  \bigoplus_{i=1}^{p} 
  \left(
 \begin{array}{cc}
  0 & -a_i \\ a_i & 0
 \end{array}
 \right),\quad\quad&&\mbox{if $n$ is odd} \label{eqFcenodd}\\
 F(\Yv) & = \left(
 \begin{array}{cc}
  0 & -\lambda^{-1/2} \\ -\lambda^{-1/2} & 0
 \end{array}
 \right)
  \bigoplus_{i=1}^{p+1} 
  \left(
 \begin{array}{cc}
  0 & -a_i \\ a_i & 0
 \end{array}
 \right),\quad\quad&&\mbox{if $n$ is even,}\label{eqFceneven}
\end{align}
where this block form is adapted to the following orthogonal decomposition of $\mathbb{M}^{1,n+1}$ as a sum of $F$-invariant subspaces
\begin{equation}\label{eqdecMink}
 \mathbb{M}^{1,n+1} = \Pi_t \oplus \spn{e_2} \bigoplus_{i=1}^p \Pi_i,\quad\mbox{($n$ odd)},\quad \quad \mathbb{M}^{1,n+1} = \Pi_t  \bigoplus_{i=1}^{p+1} \Pi_i, \quad\mbox{($n$ even)},
\end{equation}
where $\Pi_t =\spn{e_0,e_1}$ for both cases and $\Pi_i = \spn{e_{2i +1},e_{2 i+2}}$ for $n$ odd and $\Pi_i = \spn{e_{2i},e_{2 i+1}}$ for $n$ even. Any timelike or null vector $v \in \mink{1,n+1}$ must have non-zero projection onto $\Pi_t$, so it may be written $v = v_t + v_s$, with $0 \neq v_t \in \Pi_t,~v_s \in (\Pi_t)^\perp$. Hence $F(\Yv) (v) = F(\Yv) (v_t) + F(\Yv) (v_s)$, where from the block form it follows that $0 \neq F(\Yv) (v_t) \in \Pi_t$ and $F(\Yv)(v_s) \in (\Pi_t)^\perp$, thus $F(\Yv) (v) = F(\Yv) (v_t) + F(\Yv) (v_s) \neq 0$. Therefore, $\ker F(\Yv)$ is always spacelike or cero.
It is straightforward to compute the polynomial $\mathcal{Q}_{F^2}(x)$
 \begin{equation}
  \mathcal{Q}_{F^2}(x) = (x + \lambda) \prod_{i=1}^q (x - a_i^2)
 \end{equation}
 where we may order the indices $i$, so that the rotation parameters $a_i$ appear in decreasing order $a_1^2\geq \cdots \geq a_{q}^2$. Hence, applying Proposition \ref{defgammamu} we identify the parameters $\sigma := -\lambda^{-1}$ and $\mu_i^2 := a_i^2$ for $n$ odd and $-\mu_t^2 := -\lambda^{-1},~\mu_s^2 := a_1^2$ and $\mu_i^2 := a_{i+1}^2$ for $n$ even.  Therefore:
\begin{theorem}   
 Let $\tilg_{KdS}$ be a metric of the generalized Kerr-de Sitter family of metrics in all dimensions, namely given by \eqref{eqgKS} and \eqref{eqgdSodd}, \eqref{eqkodd}, \eqref{eqgHodd},  with cosmological constant $\lambda$ and $q$ rotation parameters $a_i$ sorted by $a_1^2 \geq\cdots \geq a_q^2$. Then $\tilg_{KdS}$ is uniquely characterized by the class of initial data $(\Sigma, \gamma, D_\Y)$, where $\gamma$ is conformally flat and $D_\Y$ is a TT tensor of $\gamma$ of the form \eqref{eqDKds}, where $\Yv$ (such that $\gamma(\Yv, \cdot) = \Yf$) is a CKVF of $\gamma$ whose conformal class is uniquely determined according to Proposition \ref{defgammamu} by the parameters $\{\sigma = -\lambda^{-1},\mu_1^2 = a_1^2,\cdots,\mu_p^2 = a_p^2\}$ if $n$ odd and $\{-\mu_t^2 = - \lambda^{-1},\mu_s^2 = a_1^2; \mu_1^2 = a_2^2,\cdots,\mu_p^2 = a_{p+1}^2\}$ if $n$ is even.
\end{theorem}

  We conclude this paper by comparing our results with previous literature in the $\lambda<0$ case. The metrics in \cite{Gibbons2005} admit both signs of
    $\lambda$, so one also has the family of Kerr-anti de Sitter metrics in all dimensions. The boundary metric $\gamma$ for this case is given by \eqref{eqYflat}, which is now Lorentzian. The electric part of the rescaled Weyl tensor is
    $D = \kappa_{|\lambda|} D_\Y$, where $\kappa_{|\lambda|}$ is obtained from
    $\kappa$ above by simply replacing $\lambda \rightarrow |\lambda|$, and $D_\Y$ is \eqref{eqDKds}, with $\Y$ given by \eqref{eqYflat}. These data characterize the spacetime asymptotically.

   One of the main focus on the Kerr-anti de Sitter metrics has been to study conserved quantities at infinity. There are various notions of conserved charges (see the references in \cite{Holl05}, where the different definitions are compared), but all of them depend on a CKVF $\Y$ of $\scri$. Thus, associated to each $\Y$ one defines a conserved charge $\mathcal{J}_0(\Y)$. This provides a useful method to defined mass in this context. There is no complete agreement as to which CKVF at infinity should be used to define mass. See for instance the $n=4$ cases in  \cite{papasken05} and \cite{Gibb05} or higher dimensional cases in \cite{Das00}. From our analysis, in the  Kerr-anti de Sitter case the
    boundary data itself singles out
    a  privileged CKVF, and it is most natural to use this CKVF to define the mass.
    It turns out that this CKVF agrees with the choice made in
    \cite{papasken05} for completely different reasons. It would be worth to investigate whether there is a deeper reason for this, perhaps in the context of holography.

 \appendix
 
 \section{Derivation of Weyl tensor formulae.}\label{appA}
 
 {In this Appendix we give proofs for Lemmas \ref{lemmaCt} and \ref{lemmaWeyls}. Before doing so, we list well-known identities, that we shall require in both proofs, relating the geometry of any two metrics $g$ and $\tilg$ (not necessarily conformal to each other for the moment). 
Our convention for the Riemann tensor is }
\begin{equation*}
 {R^\mu}_{\alpha \nu \beta}X_\mu = -\nabla_\nu \nabla_\beta X_\alpha + \nabla_\beta \nabla_\nu X_\alpha.
\end{equation*}
The general formula that relates the two Riemann tensors ${R^\mu}_{\alpha\nu\beta}$ and ${{\widetilde{R}}^\mu}_{~~\alpha\nu\beta}$ of $g$ and $\tilg$ respectively is (see e.g. \cite{Kroonbook}):
\begin{equation}\label{eqdiffriems}
 {R^\mu}_{\alpha\nu\beta} - {{\widetilde{R}}^\mu}_{~~\alpha\nu\beta} = 2 \nabla_{[\nu} {Q^\mu}_{\beta] \alpha} 
 + 
 2 {Q^\sigma}_{[\nu | \alpha|}{{Q^\mu}_{\beta] \sigma}}
\end{equation}
where $Q = \nabla - \widetilde \nabla$ is the difference of connections tensor
\begin{equation}\label{eqQtensdef}
 {Q^\mu}_{\alpha \beta} 
 := \frac{1}{2} \tilg^{\mu \nu}(\nabla_\nu \tilg_{\alpha \beta} - \nabla_\alpha \tilg_{\beta \nu} - \nabla_\beta \tilg_{\alpha \nu}).
\end{equation}
If they are conformally related $g = \Om^2 \tilg$, then $Q$ reads
\begin{equation}\label{eqQtens}
  {Q^\mu}_{\alpha \beta} = \frac{1}{\Om} \lr{T_\alpha {\delta^\mu}_\beta  + T_\beta {\delta^\mu}_\alpha - T^\mu g_{\alpha \beta} },\quad\quad T_\mu := \nabla_\mu \Om, \quad T^\mu:=g^{\mu \nu} T_\nu.
\end{equation}

 \subsection{Proof of Lemma \ref{lemmaCt}}\label{appelec}

 {In this subsection we prove Lemma \ref{lemmaCt}. Let $\tilg$ me a conformally extendable $(n+1)$-dimensional metric and let also $g = \Om^2 \tilg$ be a geodesic conformal extension. We define the following contraction of the Riemann tensor with $T$}
\begin{equation}\label{eqRTdef}
 \lr{R_T}_{\alpha \beta} := {R}_{\mu\alpha \nu \beta} T^\mu T^\nu 
\end{equation}
and denote 
\begin{equation}\label{eqRTdefs}
  A_{\alpha \beta}:= \nabla_\alpha T_\beta,
  \quad 
  A^2_{\alpha\beta} :=\nabla_\alpha T^\mu \nabla_\mu T_\beta.
\end{equation}
Observe that $A$ is symmetric. 
Since $T$ is geodesic, we have
\begin{align}
\lr{R_T}_{\alpha \beta} & =  T^\nu (-\nabla_\nu \nabla_\beta T_\alpha +  \nabla_\beta \nabla_\nu T_\alpha) \\ 
& = -\nabla_T \nabla_\beta T_\alpha +  \nabla_\beta \nabla_T T_\alpha -   \nabla_\beta T^\nu \nabla_\nu T_\alpha  =-\nabla_T A_{\alpha \beta}- A^2_{\alpha\beta} \label{eqRTexp}
\end{align}
Using expressions \eqref{eqdiffriems} and \eqref{eqQtens} we calculate the difference of tensors $R_T$ and $\widetilde R_T$
defined analogously for $\tilg$, specifically
\begin{equation}
 \lr{R_T}_{\alpha \beta} := {\widetilde R}_{\mu\alpha \nu \beta} T^\mu T^\nu .
\end{equation}
 First
\begin{align}
 2 T_\mu T^\nu \nabla_{[\nu} {Q^\mu}_{\beta] \alpha}  = & T_\mu T^\nu \nabla_\nu\big( \frac{1}{\Om} (T_\alpha {\delta^\mu}_\beta + T_\beta {\delta^\mu}_\alpha - g_{\alpha \beta} T^\mu)\big) \\ & -  T_\mu T^\nu \nabla_\beta\big( \frac{1}{\Om}(T_\nu {\delta^\mu}_\alpha + T_\alpha {\delta^\mu}_\nu - g_{\alpha \nu} T^\mu)\big) \\ 
  = & -\frac{1}{\Om^2}T_\mu T^\nu T_\nu (T_\alpha {\delta^\mu}_\beta + T_\beta {\delta^\mu}_\alpha - g_{\alpha \beta} T^\mu) \\ & +\frac{1}{\Om^2}T_\mu T^\nu T_\beta (T_\nu {\delta^\mu}_\alpha + T_\alpha {\delta^\mu}_\nu  - g_{\alpha \nu} T^\mu)
 -\frac{T_\mu T^\mu}{\Om} \nabla_\beta T_\alpha \\
= & \frac{\lambda}{\Om^2}(2 T_\alpha T_\beta + \lambda g_{\alpha \beta})  -\frac{\lambda}{\Om^2} T_\alpha T_\beta +  \frac{\lambda}{\Om}  \nabla_\beta T_\alpha = \frac{\lambda}{\Om} \nabla_\alpha T_\beta +\frac{\lambda}{\Om^2} \lr{T_\alpha T_\beta + \lambda g_{\alpha \beta}} 
  \end{align}
  and second
  \begin{align}
  2 T_\mu T^\nu {Q^\sigma}_{[\nu | \alpha|}{{Q^\mu}_{\beta] \sigma}} & =   \frac{T_\mu T^\nu}{\Om^2} (T_\nu {\delta^\sigma}_\alpha + T_\alpha {\delta^\sigma}_\nu - g_{\alpha \nu} T^\sigma)
  (T_\beta {\delta^\mu}_\sigma + T_\sigma {\delta^\mu}_\beta - g_{\beta \sigma} T^\mu) \\
  & - \frac{T_\mu T^\nu}{\Om^2} (T_\beta {\delta^\sigma}_\alpha +T_\alpha {\delta^\sigma}_\beta  - g_{\alpha \beta} T^\sigma)
  (T_\nu {\delta^\mu}_\sigma + T_\sigma {\delta^\mu}_\nu - g_{\nu \sigma} T^\mu) \\ 
  & = -\frac{\lambda}{\Om^2} ( 2 T_{\alpha}T_{\beta} + \lambda g_{\alpha \beta}) +  \frac{\lambda}{\Om^2} ( 2 T_{\alpha} T_{ \beta} + \lambda g_{\alpha \beta})   = 0 .
\end{align} 
Recalling $\nabla_\alpha T_\beta = A_{\alpha \beta}$, we have that \eqref{eqdiffriems} yields
\begin{equation}\label{eqdiffriemsconf} 
 ({R_T})_{\alpha \beta }- \Om^2 ( {\widetilde R_T})_{\alpha \beta} = \frac{\lambda}{\Om} A_{\alpha\beta} +\frac{\lambda}{\Om^2} \lr{T_\alpha T_\beta + \lambda g_{\alpha \beta}}.
\end{equation}
Assume now that $\tilg$ is an ACC metric in normal form w.r.t. to the boundary metric. Thus, from \eqref{eqnormform}, in Gaussian coordinates
 \begin{equation}
  T^\alpha \partial_\alpha = -\lambda  \partial_\Om.
 \end{equation}
Also, $A$, which is, up to a constant factor, the second fundamental form of the leaves $\Sigma_\Om = \lrbrace{\Om = const.}$, is
\begin{equation}\label{eqAder}
 A_{\alpha\beta} = \nabla_\alpha T_\beta = -\Gamma^0_{\alpha \beta} = \frac{g^{00}}{2}\partial_\Om g_{\alpha \beta} = -\frac{\lambda}{2} \partial_\Om g_{\alpha\beta}
\end{equation}
and its covariant derivative w.r.t. $T$ 
\begin{equation*}
 \nabla_T A_{\alpha \beta} = -\lambda \partial_{\Om}A_{\alpha \beta} + \lambda \lr{\Gamma_{\alpha 0}^\mu A_{\mu \beta} + \Gamma_{\beta 0}^\mu A_{\alpha \mu}},
\end{equation*}
 with
\begin{equation*}
 \Gamma_{\alpha 0}^\mu = \frac{1}{2} g^{\mu \nu} \lr{\partial_{\alpha} g_{0 \nu}  + \partial_\Om g_{\alpha \nu}  - \partial_{\nu} g_{0 \alpha}} = \frac{1}{2} g^{\mu \nu} \partial_\Om g_{ \alpha \nu} = -\frac{1}{ \lambda} {A^\mu}_\alpha,
\end{equation*}
so in consequence
\begin{equation}\label{eqcovTA}
 \nabla_T A = - \lambda \partial_\Om A - 2 A^2.
\end{equation}
Inserting  \eqref{eqRTexp} and \eqref{eqcovTA} in equation \eqref{eqdiffriemsconf} yields
\begin{equation}\label{eqelecriem} 
 \Om^2 ( {\widetilde R_T})_{\alpha \beta} =\lambda \partial_\Om A_{\alpha \beta} + A^2_{\alpha \beta} -  \frac{\lambda}{\Om} A_{\alpha\beta} -\frac{\lambda}{\Om^2} \lr{T_\alpha T_\beta + \lambda g_{\alpha \beta}}. 
\end{equation}
  If furthermore, we assume that $\tilg$ is Einstein with cosmological constant $\Lambda \neq 0$
\begin{equation*}
 \widetilde R_{\mu \alpha \nu \beta} = \widetilde C_{\mu \alpha \nu \beta} + 2 \lambda \widetilde g_{\mu[\nu}\widetilde g_{\beta]\alpha},
\end{equation*}
 we can relate $\widetilde R_T$ to the  $T$-electric part of the Weyl tensor \eqref{eqTelec} by 
\begin{equation}\label{eqelecweylriem}
  (\widetilde R_T)_{\alpha \beta} = \frac{\lr{C_{T}}_{\alpha \beta}}{ \Om^2} - \lambda \lr{\frac{\lambda g_{\alpha \beta} + T_\alpha T_\beta}{\Om^4}}.
\end{equation}
 Combining \eqref{eqelecriem} and \eqref{eqelecweylriem} gives
\begin{equation}\label{eqelecriem2} 
  ({ C_T})_{\alpha \beta} =\lambda \partial_\Om A_{\alpha \beta} + A^2_{\alpha \beta} -  \frac{\lambda}{\Om} A_{\alpha\beta} 
\end{equation} 
  which after writting $A_{\alpha \beta}$ in terms of the metric with \eqref{eqAder} yields the following result: 
   \begin{equation}
 \lr{C_{T}}_{ij} = \frac{\lambda^2}{2} \lr{\frac{1}{2} \partial_{\Om} g_{ik} g^{kl} \partial_{\Om} g_{lj}+  \frac{1}{\Om} \partial_\Om g_{ij} - \partial^{2}_\Om g_{ij}}.
\end{equation}
  
\subsection{Proof of Lemma \ref{lemmaWeyls}}\label{appdomweyl}

{We now give a proof for Lemma \ref{lemmaWeyls}. 
%
%
%
 Let two $(n+1)$-dimensional metrics $g$ and $\hat g$ be related by the formula \eqref{eqmetsghatg}. First notice that the inverse metrics $g^{-1}$ and $\hat g^{-1}$ must be related by a similar formula
\begin{equation}\label{eqmetinvghatg}
 g^{-1} = \hat g^{-1} + \Om^m l,
\end{equation}
for a contravariant two-tensor $l$  (also $C^2$ near $\{\Om = 0\}$, as well as $g^{-1}, \hat g^{-1}$), because the presence of any term of order $\Om^{m'},~m'<m$,  would imply terms of order $\Om^{m'}$ in $g^{-1} g$ which could not be cancelled. When using indices, we will omit the $^{-1}$ in the metrics and write upper indices. Also, indices in objects with hats are moved with the metric $\hat{g}$ and its inverse and indices of unhatted tensors are moved with $g$. 

Recall the definition of the Weyl tensor
\begin{equation}\label{eqWeyl}
 {C^\mu}_{\nu \alpha \beta} = {R^\mu}_{\nu\alpha \beta} +{\A^\mu}_{\nu \alpha \beta} \quad\quad\mbox{with}\quad \quad {\A^\mu}_{\nu \alpha \beta} := - \frac{2}{n-1}({\delta^\mu}_{[\alpha}R_{\beta]\nu} - g_{\nu[\alpha}{R^\mu}_{|\beta}) + \frac{2 R}{n(n-1)}  {\delta^\mu}_{[\alpha}g_{\beta] \nu}.
\end{equation}
Using the relation of Riemann tensors \eqref{eqdiffriems} for $g$ and $\hat g$ and \eqref{eqWeyl} we find 
\begin{equation}
  {C^\mu}_{\nu \alpha \beta} =  {{\hat C}^\mu}_{~~\nu \alpha \beta} + {\B^\mu}_{\nu \alpha \beta} + {\A^\mu}_{\nu \alpha \beta} - {\hat \A^\mu}_{~~\nu \alpha \beta}\quad\quad\mbox{with}\quad\quad{\B^\mu}_{\nu \alpha \beta} :=  2 \nabla_{[\alpha} {Q^\mu}_{\beta] \nu} 
 + 
 {Q^\sigma}_{[\alpha | \nu|}{{Q^\mu}_{\beta] \sigma}}
\end{equation}
where $Q$ is the difference of connections tensor \eqref{eqQtensdef}. We also define $\B_{\alpha \beta} = {\B^\mu}_{\alpha \mu \beta}$ and $\B = g^{\alpha \beta}\B_{\alpha \beta}$ so that
\begin{equation}
 R_{\alpha \beta} - \hat R_{\alpha \beta} = \B_{\alpha \beta},\quad\quad 
 {R^\mu}_\beta - {\hat R^\mu}_{~\beta} = {\B^\mu}_\beta + \Om^m l^{\mu \alpha} \hat R_{\alpha \beta},
\quad\quad R- \hat R = \B + \Om^m l^{\mu \beta} \hat R_{\mu \beta}.
\end{equation}
With these definitions we expand ${\A^\mu}_{\nu \alpha \beta}$
\begin{align}
 {\A^\mu}_{\nu \alpha \beta} & = -\frac{2}{n-1}\lr{{\delta^\mu}_{[\alpha}{\hat R_{\beta]\nu}} - \hat g_{\nu[\alpha} \hat R^\mu_{~~\beta]}} + \frac{2}{n(n-1)} \hat R {\delta^\mu}_{[\alpha}\hat g_{\beta]\nu}\\
 & -\frac{2}{n-1} \lr{{\delta^\mu}_{[\alpha} \B_{\beta] \nu} -\hat g_{\nu [\alpha} {B^\mu}_{\beta]} - \Om^m\lr{\hat g_{\nu [\alpha}  \hat R_{\beta] \sigma} l^{\mu \sigma} + q_{\nu [\alpha} {\hat R^\mu}_{~\beta]} + q_{\nu [\alpha} {B^\mu}_{\beta]} } - \Om^{2m}q_{\nu [\alpha}{\hat R}_{\beta]\sigma}l^{\sigma \mu}   }\\
 & + \frac{2B}{n(n-1)} {\delta^\mu}_{[\alpha}\hat g_{\beta]\nu} + \frac{2 \Om^m}{n (n-1)}\lr{l^{\lambda \sigma} \hat R_{\lambda \sigma}{\delta^\mu}_{[\alpha}\hat g_{\beta]\nu} +(\hat R + \B){\delta^\mu}_{[\alpha}q_{\beta]\nu}} + \frac{2 \Om^{2m}}{n (n-1)}  l^{\lambda \sigma} \hat R_{\lambda \sigma}{\delta^\mu}_{[\alpha}q_{\beta]\nu},
\end{align}
so defining
\begin{align}
 {G^\mu}_{\nu \alpha \beta} & : = -\frac{2}{n-1} \lr{{\delta^\mu}_{[\alpha} \B_{\beta] \nu} -\hat g_{\nu [\alpha} {B^\mu}_{\beta]} - \Om^m\lr{\hat g_{\nu [\alpha}  \hat R_{\beta] \sigma} l^{\mu \sigma} + q_{\nu [\alpha} {\hat R^\mu}_{~\beta]} + q_{\nu [\alpha} {B^\mu}_{\beta]} } - \Om^{2m}q_{\nu [\alpha}{\hat R}_{\beta]\sigma}l^{\sigma \mu}   }\\
& + \frac{2B}{n(n-1)} {\delta^\mu}_{[\alpha}\hat g_{\beta]\nu} + \frac{2 \Om^m}{n (n-1)}\lr{l^{\lambda \sigma} \hat R_{\lambda \sigma}{\delta^\mu}_{[\alpha}\hat g_{\beta]\nu} +(\hat R + \B){\delta^\mu}_{[\alpha}q_{\beta]\nu}} + \frac{2 \Om^{2m}}{n (n-1)}  l^{\lambda \sigma} \hat R_{\lambda \sigma}{\delta^\mu}_{[\alpha}q_{\beta]\nu}.
\end{align}
gives 
\begin{equation}
{\A^\mu}_{\nu \alpha \beta} = {{\hat \A}^\mu}_{~~\nu \alpha \beta} +{G^\mu}_{\nu \alpha \beta},
\end{equation}
from which
\begin{equation}\label{eqintweyl}
  {C^\mu}_{\nu \alpha \beta} =  {{\hat C}^\mu}_{~~\nu \alpha \beta} + {\B^\mu}_{\nu \alpha \beta} + {G^\mu}_{\nu \alpha \beta}.
\end{equation}
We now analyze the behaviour near $\{\Om = 0 \}$ of the tensors $\B$ and $G$. 
Using formula \eqref{eqQtensdef} (with $\tilg \rightarrow \hat g = g - \Om^m q$ ) we have
\begin{align}
 \hat Q_{\nu \alpha \beta}  := \hat g_{\mu \nu}{Q^\mu}_{\alpha \beta} 
 & =  -F \frac{m}{2} \Om^{m-1}   \lr{\nor_\nu q_{\alpha \beta} - \nor_\alpha q_{\beta \nu} - \nor_\beta q_{\alpha \nu}}  - \frac{\Om^m}{2} \lr{\nabla_\nu q_{\alpha \beta} - \nabla_\alpha q_{\beta \nu} - \nabla_\beta q_{\alpha \nu}} \\
 & =  -  F \frac{m}{2} \Om^{m-1}   \lr{\nor_\nu q_{\alpha \beta} - \nor_\alpha q_{\beta \nu} - \nor_\beta q_{\alpha \nu}} + O(\Om^{m}) = O(\Om^{m-1}).
\end{align}
On the other hand 
\begin{align}
 \nabla_\mu \hat Q_{\nu \alpha \beta} & = - F^2 \frac{m(m-1)}{2} \Om^{m-2}  u_\mu \lr{\nor_\nu q_{\alpha \beta} - \nor_\alpha q_{\beta \nu} - \nor_\beta q_{\alpha \nu}} - \frac{\Om^m}{2}\nabla_\mu \lr{\nabla_\nu q_{\alpha \beta} - \nabla_\alpha q_{\beta \nu} - \nabla_\beta q_{\alpha \nu}} \\
 & - m \frac{\Om^{m-1}}{2} \lr{ \nabla_\mu \big(F  \lr{\nor_\nu q_{\alpha \beta} - \nor_\alpha q_{\beta \nu} - \nor_\beta q_{\alpha \nu}}\big) +  F \nor_\mu \lr{\nabla_\nu q_{\alpha \beta} - \nabla_\alpha q_{\beta \nu} - \nabla_\beta q_{\alpha \nu}}}\\
 & =  - F^2 \frac{m(m-1)}{2} \Om^{m-2} u_\mu   \lr{\nor_\nu q_{\alpha \beta} - \nor_\alpha q_{\beta \nu} - \nor_\beta q_{\alpha \nu}} + O(\Om^{m-1})
\end{align}
thus
\begin{align}
 \nabla_\mu {Q^\nu}_{\alpha \beta} & =  \nabla_\mu (\hat g^{\sigma \nu}{\hat Q}_{\sigma \alpha \beta}) =  \nabla_\mu \big( (g^{\sigma \nu} - \Om^m l^{\sigma \nu}){\hat Q}_{\sigma \alpha \beta}\big) = g^{\sigma \nu} \nabla_{\mu} \hat Q_{\sigma \alpha \beta} + O(\Om^{m-1})\\
 & = -F^2 \frac{m (m-1)}{2}\Om^{m-2} u_\mu \lr{u^\nu q_{\alpha \beta} - u_\alpha {q^\nu}_\beta - u_\beta {q^\nu}_\alpha} + O(\Om^{m-1}).
\end{align}
Therefore, the leading order terms of $\B$ are
\begin{align}
 {\B^\mu}_{\nu \alpha \beta} = 2 \nabla_{[\alpha}{Q^\mu}_{\beta]\nu} + O(\Om^{2m-2}) =  - m(m-1) F^2 \Om^{m-2} \big(\nor^\mu \nor_{[\alpha}q_{\beta] \nu} + {q^\mu}_{[\alpha} \nor_{\beta]} \nor_\nu\big)+ O(\Om^{m-1}) = O(\Om^{m-2}). 
\end{align}
 
Next, we calculate the leading order terms of $G$. Notice that since $\hat g$ is $C^2$ at $\{ \Om = 0\}$, its Ricci tensor is well-defined.  Moreover $\B$ and all its traces are $O(\Om^{m-2})$. Thus
\begin{equation}
 {G^\mu}_{\nu \alpha \beta} = -\frac{2}{n-1}\big( {\delta^\mu}_{[\alpha}\B_{\beta]\nu} - \hat g_{\nu[\alpha} {\B^\mu}_{\beta]} \big) + \frac{2 \B}{n(n-1)} {\delta^\mu}_{[\alpha}\hat g_{\beta]\nu} + O(\Om^m).
\end{equation}

If $\nor$ is non-null, i.e. $\epsilon \neq 0$, it is useful to decompose $q$ in terms parallel and orthogonal to $\nor$, i.e.
\begin{equation}
 q_{\alpha\beta} = U \nor_\alpha \nor_\beta + 2 \nor_{(\alpha} V_{\beta)} + t_{\alpha \beta},\quad\quad\mbox{with}\quad \nor^\mu V_\mu = 0,~\nor^\mu t_{\mu\nu} = 0.
\end{equation}
Similarly, the following decomposition of the metric holds
\begin{equation}\label{eqperplong} 
 g_{\alpha\beta} = \epsilon \nor_\alpha \nor_\beta + h_{\alpha \beta}
\end{equation}
which defines $h_{\alpha\beta}$ as the projector orthogonal to $\nor$. In terms of these quantities
\begin{equation}\label{eqBdec}
{\B^\mu}_{\nu \alpha \beta} = - m(m-1) \Om^{m-2} F^2 \big( 
\nor^\mu \nor_{[\alpha}t_{\beta] \nu} + {t^\mu}_{[\alpha} \nor_{\beta]}\nor_\nu
\big) + O(\Om^{m-1})
\end{equation} 
and 
\begin{align}
 B_{\beta \nu} & =  {B^\mu}_{\beta \mu \nu} = - \frac{1}{2} m(m-1) \Om^{m-2}F^2(\epsilon t_{\beta \nu} + t u_\beta u_\nu) + O(\Om^{m-1})\\
 B & =  {B^\mu}_{\mu} = - \frac{1}{2} m(m-1) \Om^{m-2} F^2 (2 \epsilon t) + O(\Om^{m-1})
\end{align}
where $t = g^{\alpha\beta}t_{\alpha\beta} = h^{\alpha\beta} t_{\alpha\beta}$. 
In consequence, 
\begin{align}
 {G^\mu}_{\nu \alpha \beta} & = -  m(m-1) \Om^{m-2} F^2 \times \Big( 
  \frac{-1}{n-1} \big(\epsilon {\delta^\mu}_{[\alpha}t_{\beta]\nu} + t {\delta^\mu}_{[\alpha}\nor_{\beta]} \nor_\nu - \epsilon \hat g_{\nu [\alpha}{t^\mu}_{\beta ]}-  t \nor^\mu \hat g_{\nu [\alpha}\nor_{\beta]} \big)\\ 
  & + \frac{2 \epsilon t}{n(n-1)} {\delta^\mu}_{[\alpha}\hat g_{\beta] \nu}
 \Big) + O(\Om^{m-1}) \label{eqDdec}
\end{align}
From \eqref{eqperplong} one has 
\begin{equation}
\hat g_{\alpha \beta}= \epsilon \nor_\alpha \nor_\beta + h_{\alpha \beta} + O(\Om^m),\quad\quad {\delta^\alpha}_\beta = \epsilon u^\alpha u_\beta + {h^\alpha}_\beta,
\end{equation}
so that \eqref{eqDdec} reads
\begin{align}
 {G^\mu}_{\nu \alpha \beta} & = - m(m-1) \Om^{m-2} F^2 \times \Big(
  \frac{-1}{n-1} \big(u^\mu u_{[\alpha} t_{\beta]\nu} + \epsilon {h^\mu}_{[\alpha}t_{\beta] \nu} +t {h^\mu}_{[\alpha} u_{\beta]} u_\nu + {t^\mu}_{[\alpha}u_{\beta]} u_\nu  \\ 
  & + \epsilon {t^\mu}_{[\alpha} h_{\beta] \nu} + t u^\mu u_{[\alpha}h_{\beta]\nu}\big) 
   + \frac{2  t}{n(n-1)} \lr{u^\mu u_{[\alpha} h_{\beta] \nu} + {h^\mu}_{[\alpha} u_{\beta]} u_\nu + \epsilon {h^\mu}_{[\alpha} h_{\beta]\nu} }
 \Big) + O(\Om^{m-1})    \\
 & = -  m(m-1) \Om^{m-2} F^2 \times \Big( 
  \frac{-1}{n-1} \big(u^\mu u_{[\alpha} t_{\beta]\nu} + {t^\mu}_{[\alpha}u_{\beta]} u_\nu  \big) - t \frac{n-2 }{n(n-1)} \lr{u^\mu u_{[\alpha} h_{\beta] \nu} + {h^\mu}_{[\alpha} u_{\beta]} u_\nu }\\ 
  & - \frac{\epsilon}{n-1}\big( {h^\mu}_{[\alpha}t_{\beta] \nu} -  \frac{t}{n}{h^\mu}_{[\alpha} h_{\beta]\nu} +  {t^\mu}_{[\alpha} h_{\beta] \nu} -  \frac{t}{n}{h^\mu}_{[\alpha} h_{\beta]\nu} \big)    
 \Big) + O(\Om^{m-1}).   \label{eqDdec2} 
\end{align}
 Denote the traceless part of $t_{\alpha \beta}$  by
\begin{equation}
 \tlt_{\alpha \beta} = t_{\alpha \beta} - \frac{t}{n}h_{\alpha\beta}.
\end{equation}
Also, notice that the lower order terms of all expression are $O(\Om^{m-1}) = o(\Om^{m-2})$ for $m\geq 2$. Hence, combining \eqref{eqintweyl}, \eqref{eqBdec} and \eqref{eqDdec2} gives the following result 
\begin{equation}
 {C^\mu}_{\nu\alpha\beta } = {\hat C^\mu}_{~~\nu\alpha\beta } -  K_m(\Om) \frac{n-2}{n-1}(\nor^\mu \nor_{[\alpha} \tlt_{\beta]\nu} + {\tlt^\mu}_{~~ [\alpha}\nor_{\beta]}\nor_\nu) + \frac{\epsilon K_m(\Om)}{n-1}({h^\mu}_{[\alpha} \tlt_{\beta]\nu} + \tlt^\mu_{~~ [\alpha} h_{\beta]\nu}) + o(\Om^{m-2}) .
\end{equation}

%
%

\section*{Acknowledgements}

The authors acknowledge financial support under the projects
PGC2018-096038-B-I00
(Spanish Ministerio de Ciencia e Innovaci\'on and FEDER ``A way of making Europe'')
and SA096P20 (JCyL). C. Pe\'on-Nieto also acknowledges the Ph.D. grant BES-2016-078094 (Spanish Ministerio de Ciencia e Innovaci\'on).


\begin{thebibliography}{10}

\bibitem{Anderson2004}
M.~T. Anderson.
\newblock On the structure os asymptotically de Sitter and anti-de Sitter
  spaces.
\newblock {\em Adv. Theor. Math. Phys.}, {\bf 8}:861--893, 2004.
\bibitem{Anderson2005}
M.~T. Anderson.
\newblock Existence and stability of even-dimensional asymptotically de
  {Sitter} spaces.
\newblock {\em Ann. Henri Poincar{\'e}}, {\bf 6}:801--820, 2005.

\bibitem{andersonchrusciel05}
M.~T. Anderson and P.~T. Chruściel.
\newblock Asymptotically simple solutions of the vacuum einstein equations in
  even dimensions.
\newblock {\em Comm. Math. Phys.}, {\bf 260}:557–577,
  2005. 

 \bibitem{barnes04}
 A.~ Barnes.
 \newblock {Purely magnetic space-times}.
 \newblock  {\em {27th Spanish Relativity Meeting (ERE 2003): Gravitational
   Radiation} (Alicante, September 2003)}, editors J.~A.~Miralles, J.~A.~Font and J.~A.~Pons, Publicacions Universitat d'Alacant, Alicante, 2004.

   \bibitem{blairconformal}
D.~E. Blair.
\newblock \emph{Inversion Theory and Conformal Mapping}, volume~9 of
  \emph{Student Mathematical Library}.
\newblock American Mathematical Society, {Providence, Rhode Island}, 2000.
   
\bibitem{Das00}
S.~Das and R.~B. Mann.
\newblock Conserved quantities in Kerr-anti-de Sitter spacetimes in various
  dimensions.
\newblock {\em JHEP}, {\bf 2000}:033, 2000.

 \bibitem{eastwoodrev}
 M.~Eastwood.
 \newblock Book review-the ambient metric by C. Fefferman \& R. Graham.
 \newblock {\em Bull. Amer. Math. Soc.}, {\bf 51}, 499--503, 2014.

 \bibitem{evanspdes}
 L.~C. Evans.
 \newblock {\em Partial Differential Equations}.
 \newblock Graduate studies in mathematics. American Mathematical Society, Providence, 2010.

 \bibitem{FeffGrah85}
 C.~Fefferman and C.~R. Graham.
 \newblock Conformal invariants.
 \newblock {\em \'Elie Cartan et les math\'ematiques d'aujourd'hui (Lyon,
   25-29 juin 1984)}, Ast\'erisque, 95--116. Soci\'et\'e
   math\'ematique de France, 1985.
 
 \bibitem{ambientmetric}
 C.~Fefferman and C.~R. Graham.
 \newblock {\em The Ambient Metric (AM-178)}.
 \newblock Princeton University Press, 2012.

 \bibitem{friedrich81bis}
 H.~Friedrich.
 \newblock {On the Regular and Asymptotic Characteristic Initial Value Problem
   for Einstein's Vacuum Field Equations}.
 \newblock {\em Proc. Roy. Soc. Lond. A}, {\bf 375}, 169--184, 1981.
 
 \bibitem{friedrich81}
 H.~Friedrich.
 \newblock {The asymptotic characteristic initial value problem for Einstein's
   vacuum field equations as an initial value problem for a first-order
  quasilinear symmetric hyperbolic system}.
\newblock {\em Proc. Roy. Soc. Lond. A}, {\bf 378}, 401--421, 1981.
 
 \bibitem{Fried86initvalue}
 H.~Friedrich.
 \newblock Existence and structure of past asymptotically simple solutions of
   {Einstein's} field equations with positive cosmological constant.
\newblock {\em J. Geom. Phys.}, {\bf 3}, 101 -- 117, 1986.
 
 \bibitem{Friedrich2002}
 H.~Friedrich.
 \newblock Conformal Einstein Evolution, 
 \newblock {\em The Conformal Structure of Spacetime: Geometry, Analysis, Numerics}, eds. J.~Frauendiener and H.~Friedrich,  Springer, Lecture Notes in Physics, Berlin, 2002.
 
 
 
 \bibitem{Friedrich2014}
 H.~Friedrich.
 \newblock Geometric asymptotics and beyond.
 \newblock {\em Surveys in Differential Geometry}, {\bf 20}, 37 -- 74, 2015.



 \bibitem{gasperinwilliams17}
 E. Gasperin and J.~L. Williams.
 \newblock The conformal Killing spinor initial data equations.
 \newblock ArXiv:1704.07586, 2017.  



\bibitem{Gibb05}
G.~W. Gibbons, M.~J. Perry, and C.~N. Pope.
\newblock The first law of thermodynamics for Kerr–anti-de Sitter black
  holes.
\newblock {\em Class. Quantum Grav.}, {\bf 22}:1503–1526, 2005.

 \bibitem{Gibbons2005}
 G.~W. Gibbons, H.~Lü, D.~N. Page and C.~N. Pope.
 \newblock The general {Kerr–de Sitter} metrics in all dimensions.
 \newblock {\em J. Geom. Phys.}, {\bf 53}, 49 -- 73, 2005.

 \bibitem{GrahamLee91}
 C.~R. Graham and J.~M. Lee.
 \newblock Einstein metrics with prescribed conformal infinity on the ball.
 \newblock {\em Adv. Math.}, {\bf 87}, 186 -- 225, 1991.
 
\bibitem{deHaro}
S.~de Haro, K.~Skenderis and S.~N.~ Solodukhin.
\newblock Holographic Reconstruction of spacetime and renormalization in the AdS/CFT correspondence.
\newblock {\em Comm. Math. Phys.}, {\bf 217}:595–622, 2001.

\bibitem{Hawk99}
S.~W. Hawking, C.~J. Hunter, and M.~M. Taylor-Robinson.
\newblock Rotation and the AdS-CFT correspondence.
\newblock {\em Physical Review D}, {\bf 59}, 1999.

 \bibitem{hervik13}
 S.~ Hervik, M.~ Ortaggio and L.~ Wylleman.
 \newblock Minimal tensors and purely electric or magnetic spacetimes of
   arbitrary dimension.
 \newblock {\em Class. Quantum Grav.}, {\bf 30}, 65014, 2013.

\bibitem{Holl05}
S.~Hollands, A.~Ishibashi, and D.~Marolf.
\newblock Comparison between various notions of conserved charges in
  asymptotically AdS spacetimes.
\newblock {\em Class. Quantum Grav.}, {\bf 22}:2881–2920, 2005.

\bibitem{kaminski21}
W.~Kamiński.
\newblock Well-posedness of the ambient metric equations and stability of even
  dimensional asymptotically de \mbox{Sitter} spacetimes, 2021.
\newblock 
{ArXiv:2108.08085v1 [gr-qc]}


 \bibitem{kichenassamy03}
 S. Kichenassamy.
 \newblock On a conjecture of Fefferman and Graham.
 \newblock {\em Adv. Math.}, 1{\bf 84}, 268--288, 2004.
 
 \bibitem{KdSlike}
 M.~Mars, T.~T.~Paetz and J.~M.~M.~Senovilla.
 \newblock {Classification of Kerr–de Sitter-like spacetimes with conformally
   flat {$\mathscr{I}$}}.
 \newblock {\em Class. Quantum Grav.}, {\bf 34}, 095010, 2017.
 
 \bibitem{KdSnullinfty}
 M.~Mars, T.~T.~Paetz and J.~M.~M. Senovilla and W.~ Simon.
 \newblock Characterization of (asymptotically) {Kerr–de Sitter-like}
   spacetimes at null infinity.
 \newblock {\em Class. Quantum Grav.}, {\bf 33}, 155001, 2016.

 \bibitem{marspeon21}
 M.~ Mars and C.~Peón-Nieto.
 \newblock Skew-symmetric endomorphisms in $\mathbb{M}^{1,n}$: A unified
   canonical form with applications to conformal geometry,
   \newblock {\it Class. and Quantum Grav.}. {\bf 38}, 125009, 2020.
   
   \bibitem{marspeonKSKdS21}
M.~Mars and C.~Pe\'on-Nieto.
\newblock {Classification of Kerr-de Sitter-like spacetimes with conformally
  flat $\mathscr{I}$ in all dimensions}, 
\newblock ArXiv:2109.08531 [gr-qc], 2021.
   
   \bibitem{marsseno15}
M.~Mars and J.~M.~M. Senovilla.
\newblock {A spacetime characterization of the Kerr-NUT-(A)de Sitter and
  related metrics}.
\newblock \emph{Ann. Henri Poincaré}, { \bf 16}:\penalty0 1509--1550, 2015.

 \bibitem{mazzeo88}
 R.~Mazzeo.
 \newblock {The Hodge cohomology of a conformally compactmetric}.
 \newblock {\em J. Diff. Geom.}, {\bf 28}, 309 -- 339, 1988.
 
 \bibitem{mcintosh94}
 C.~B.~G. McIntosh, R.~Arianrhod, S.~T. Wade, and C.~Hoenselaers.
 \newblock Electric and magnetic Weyl tensors: classification and analysis.
 \newblock {\em Class. Quantum Grav.}, {\bf 11}, 1555--1564, 1994.
 
 \bibitem{OrtaggioWeyl}
M.~Ortaggio and A.~Pravdová.
\newblock {Asymptotic behavior of the Weyl tensor in higher dimensions}.
\newblock \emph{Physical Review D}, { \bf 90}: 104011, 2014.
 
 \bibitem{KIDPaetz}
 T.-T. Paetz.
 \newblock {Killing Initial Data} on spacelike conformal boundaries.
 \newblock {\em J. Geom. Phys.}, {\bf 106}, 51 -- 69, 2016.

\bibitem{papasken04bis}
I.~Papadimitriou and K.~Skenderis.
\newblock Correlation functions in holographic {RG} flows.
\newblock {\em JHEP}, {\bf 2004}:075, 2004.

\bibitem{papasken04}
I.~Papadimitriou and K.~Skenderis.
\newblock {AdS / CFT correspondence and geometry}.
\newblock {\em IRMA Lect. Math. Theor. Phys.}, {\bf 8}:73--101, 2005.

\bibitem{papasken05}
I.~Papadimitriou and K.~Skenderis.
\newblock Thermodynamics of asymptotically locally AdS spacetimes.
\newblock {\em JHEP}, {\bf 2005}:004, 2005.

 \bibitem{penrose64}
 R.~Penrose.
 \newblock {Conformal treatment of infinity}.
 \newblock  {\it  Gen. Relativ. Gravit.}, {\bf 43}, 901--922, 2011 (reprint). 
 
 \bibitem{penrose65}
 R.~Penrose.
 \newblock {Zero rest mass fields including gravitation: Asymptotic behavior}.
 \newblock {\em Proc. Roy. Soc. Lond. A}, {\bf 284}, 159, 1965.
 
 \bibitem{penrose63}
 R.~ Penrose.
 \newblock Asymptotic properties of fields and space-times.
 \newblock {\em Phys. Rev. Lett.}, {\bf 10}, 66--68, 1963


 \bibitem{Rendall03}
 A.~D. Rendall.
 \newblock {Asymptotics of solutions of the Einstein equations with positive
   cosmological constant}.
 \newblock {\em Ann. Henri Poincaré}, {\bf 5}, 1041--1064, 2004.

 \bibitem{IntroCFTschBook}
 M.~Schottenloher.
 \newblock {\em A Mathematical Introduction to Conformal Field Theory}.
 \newblock Lecture Notes in Physics. Springer Berlin Heidelberg, 2008.

\bibitem{Sken02}
K.~Skenderis.
\newblock Lecture notes on holographic renormalization.
\newblock {\em Class. Quantum Grav.}, {\bf 19}:5849--5876, 2002.

 \bibitem{skenderis}
 K.~Skenderis and S.~N. Solodukhin.
 \newblock Quantum effective action from the AdS/CFT correspondence.
 \newblock {\em Phys. Lett. B}, {\bf 472}, 316 -- 322, 2000.


\bibitem{Starob82}
A.~A. Starobinsky.
\newblock {Isotropization of arbitrary cosmological expansion given an
  effective cosmological constant}.
\newblock {\em JETP Lett.}, {\bf 37}:66--69, 1983.

 
 \bibitem{Kroonbook}
 J.~A. Valiente~Kroon.
 \newblock {\em {Conformal methods in general relativity}}.
 \newblock Cambridge University Press, Cambridge, 2016.

\end{thebibliography}

\end{document}